\def\llncs{0}

\def\comments{1}

\ifnum\llncs=1
\documentclass{llncs}
\usepackage{amsmath,amssymb}
\else
\documentclass[11pt]{article}
\usepackage{amsmath,amssymb,amsthm}
\fi
\usepackage{enumitem}
\ifnum\llncs=0
\usepackage[a4paper,margin=1in]{geometry}
\fi
\usepackage{braket}
\usepackage{xcolor}

\usepackage{header}

\def\bra#1{\mathinner{\langle{#1}|}}
\def\ket#1{\mathinner{|{#1}\rangle}}
\def\braket#1{\mathinner{\langle{#1}\rangle}}
\def\ketbra#1#2{\mathinner{|{#1}\rangle\!\langle{#2}|}}

\usepackage[dvipsnames]{xcolor}
\ifnum\comments=1
\newcommand{\ynote}[1]{%
  \textcolor{purple}{\textbf{[Yiming:} #1\textbf{]}}%
}

\newcommand{\bnote}[1]{%
  \textcolor{brown}{\textbf{[Boyang:} #1\textbf{]}}%
}

\newcommand{\nnote}[1]{%
  \textcolor{blue}{\textbf{[Nikos:} #1\textbf{]}}%
}

\newcommand{\anote}[1]{%
  \textcolor{red}{\textbf{[Andrea:} #1\textbf{]}}%
}
\else
\newcommand{\bnote}[1]{}
\newcommand{\ynote}[1]{}
\newcommand{\nnote}[1]{}
\fi

\newcommand{\Gen}{\mathrm{Gen}}

\usepackage{authblk}
\usepackage{orcidlink}

\ifnum\llncs=0
\usepackage{titling}
\predate{}
\postdate{}
\fi

\title{No linear stretching of single-copy quantum pseudorandomness}
\ifnum\llncs=0
\author[1]{Boyang Chen}
\author[2]{Andrea Coladangelo}
\author[3]{Yao-Ting Lin}
\author[2]{Nikos Skoumios}
\author[2]{Justin Tysdal}
\author[2]{Yiming Wang}

\affil[1]{Department of Computer Science and Technology, Tsinghua University}
\affil[2]{Paul G. Allen School of Computer Science \& Engineering, University of Washington}
\affil[3]{Department of Computer Science, Columbia University}
\date{}
\else
\author{}
\date{}
\fi
\begin{document}
\maketitle

\begin{abstract}
Pseudorandom states, introduced by Ji, Liu, and Song (CRYPTO 2018), are quantum analogues of classical pseudorandom generators. A fundamental property of classical pseudorandom generators is that their output can be stretched to arbitrary polynomial length. Whether an analogous stretching property holds for quantum pseudorandom states has been elusive. Here, we essentially resolve this question for the single-copy variant of pseudorandom states, 1PRS. We show, via a unitary oracle separation, that no fully black-box construction can stretch the output length of a 1PRS by any constant factor greater than one.
\end{abstract}

\medskip
\medskip
\medskip
\medskip
\begin{center}
\setlength{\fboxsep}{8pt}
\fbox{%
  \parbox{0.90\textwidth}{%
    \small
    \textbf{AI usage statement.}
   This paper improves over the first version of this work (which appears as v1 at the same arXiv identifier). The first version proves an oracle separation with worse parameters (only ruling out quadratic stretching) and relative to a quantum channel oracle. That first version was produced entirely by the authors without AI. The authors then explored improvements using ChatGPT 6 Astra, which provided key ideas to upgrade the separation to rule out any linear stretching, and to make it relative to a unitary oracle. While most of the ideas behind the upgrade came from ChatGPT, the general lifting technique in Section~\ref{sec: lifting theorem} was produced by the authors. The authors also put substantial work into the writing of the final manuscript, with help from ChatGPT 6 Astra.
  }%
}
\end{center}

\ifnum\llncs=0
\newpage
\tableofcontents

\newpage
\fi

\section{Introduction}

One of the foundational challenges in cryptography is to understand the minimal assumptions under which cryptographic primitives can exist. In the classical world, this landscape is relatively well-understood: the existence of One-Way Functions (OWF) is regarded as the ``minimal'' computational assumption. It implies the existence of pseudorandom generators (PRGs), digital signatures, commitments, and symmetric-key encryption. Conversely, the existence of virtually any cryptographic primitive with computational security implies the existence of OWFs.

The quantum setting, however, presents a starkly different picture. We now know of a variety of quantum cryptographic primitives that are sufficient for building secure protocols but are \emph{potentially weaker} than OWFs. To capture this landscape, Morimae coined the term \emph{Microcrypt}, extending Impagliazzo's five worlds to include a world where quantum cryptography exists even if classical OWFs do not. One of the central primitives in Microcrypt is the \emph{pseudorandom state generator} (PRS), first introduced by Ji, Liu, and Song~\cite{JLS}. These can be thought of as the quantum analogue of PRGs. A PRS is a family of efficiently generatable quantum states indexed by a classical key, which are computationally indistinguishable from Haar random states, given polynomially many copies. Following~\cite{JLS}, numerous other primitives populating Microcrypt have been introduced, including pseudorandom function-like states ($\mathsf{PRFS}$)~\cite{AGQY22}, EFI pairs (efficiently samplable statistically far-but-computationally-indistinguishable quantum states)~\cite{Yan22,BCQ23}, one-way state generators ($\mathsf{OWSG}$)~\cite{MY22b-OWSG}, and one-way puzzles ($\mathsf{OWPuzz}$)~\cite{KT23,CGG-EV-OWPuzz}. Most of these primitives have been shown to be qualitatively weaker than OWFs (in the formal sense of a black-box separation), starting with a result of Kretschmer~\cite{Kre21}.

In this work, we focus on \emph{quantum pseudorandomness}. The concept of pseudorandomness is ubiquitous in classical cryptography and complexity theory. Similarly, it is becoming clear that quantum pseudorandomness plays a central role not only in quantum cryptography, but also beyond, thanks to connections with fundamental physics~\cite{bouland2019complexity,brandao2021models,kim2023complementarity,engelhardt2024cryptographic,brakerski2023black}. Here, we focus more concretely on the notion of \emph{single-copy secure} pseudorandom states (1PRS), formalized by Morimae and Yamakawa~\cite{MY22a}. The notion of a 1PRS is analogous to that of PRS, the only difference being that indistinguishability from a Haar random state is only required to hold when the adversary is given a \emph{single} copy of a state from the family. It is easy to see that, in order for this notion to be non-trivial, the output length $m(n)$ of states from the family must satisfy $m(n) > n$, where $n$ is the key size. In other words, a 1PRS must be ``stretching''. 1PRS are arguably one of the most fundamental primitives in Microcrypt, as they are ``very close'' to being minimal: recent work by Cavalar et al.~\cite{CCC+26} established that a \emph{non-uniform} version of 1PRS is equivalent to EFI pairs, which are typically regarded as the minimal assumption required for quantum cryptography; and standard (i.e.\ uniform) 1PRS are implied by all other primitives mentioned above.

\paragraph{The Challenge of Stretching Quantum Pseudorandomness.} 
Many basic properties of \emph{classical} pseudorandomness have been known since the early days and are now taken for granted (e.g., PRGs can be ``shrunk'' and ``stretched''). \emph{Quantum} pseudorandomness seems to be, in some respects, much more subtle. For example, while shrinking is trivial for a classical PRG (one can just discard bits of a pseudorandom string, and obtain a shorter pseudorandom string), this is not the case for a PRS: discarding one qubit of a pseudorandom state in general yields a \emph{mixed} state! This barrier was formalized recently: work of Bouaziz--Ermann and Muguruza~\cite{BM24}, as well as the combination of concurrent work by Chen et al.~\cite{CCS25} with work by Barhoush et al.~\cite{BBO+bot-OWF}, showed that the output of a PRS cannot be shrunk in a black-box~way.

What about stretching? In the classical setting, any PRG with a non-trivial stretch can be iteratively applied to extend the output to any arbitrary polynomial length.
However, this iterative paradigm fails right away in the quantum setting: the output of a 1PRS is a quantum state, not a classical bit string, so one cannot simply use the output state as a ``seed'' for the next iteration.  Surprisingly, very little is known about whether stretching quantum pseudorandomness is possible in general or not. 

Here are some of the known results. Gunn et al.~\cite{GJMZ23} showed how to stretch the output of a PRS by an arbitrary polynomial, but at the cost of reducing the number of copies available to the adversary. So, to get 1PRS with long stretching, one needs to start with a multi-copy secure PRS, which is a qualitatively stronger object compared to 1PRS~\cite{CCS25,AGL24}. Cui et al.~\cite{Schuster-new} and Levy and Vidick~\cite{PRSlengthexpansion} introduced techniques to ``glue'' the generators of phase states together, obtaining stretching constructions for restricted types of PRS, namely (binary or general) phase PRS. However, their techniques rely heavily on the structure of the state families, and do not seem applicable to general unstructured PRS. Zhandry~\cite{Zha25-model-unitary-oracles} showed that if a 1PRS satisfies certain \emph{anti-correlation} and \emph{junk-free} properties, then its output length can be generically extended by one qubit. Schuster, Haferkamp, and Huang~\cite{SHHscience} showed that, by ``gluing'' pseudorandom unitaries, one can produce a larger pseudorandom unitary. However, their proof extensively used properties of Haar random \emph{unitaries} which do not appear applicable to pseudorandom \emph{states}. Finally, on the negative side, Bouaziz--Ermann et al.~\cite{BHMV25} show the impossibility of black-box PRS stretching for constructions with certain structural restrictions, obtaining results incomparable to ours.\footnote{Specifically, they consider stretching for standard (multi-copy) PRS, starting from a shorter-output (multi-copy) PRS—a (potentially) strictly stronger primitive than an 1PRS of the same length.}

For the general setting without any restrictions, the question has remained elusive, and the only known result is that 1PRS can be stretched in a black-box way by an additive $n^\kappa$, for any $\kappa<1$~\cite{CCC+26} (i.e.\ going from a seed of length $n$ to a state of length $n+n^\kappa$). Thus, it remains a wide open question whether quantum pseudorandomness can be stretched to arbitrary polynomial length, like its classical counterpart, or even just linearly. This leads us to the central question that we study in this work:
\begin{center}
\emph{Can 1PRS be stretched to arbitrary polynomial length in a black-box way?}
\end{center}

\subsection{Our Results}
Our main result is a negative answer to the above question. We prove that there is no black-box construction that can stretch the output length of a 1PRS by any constant factor greater than~$1$. This essentially settles the landscape of black-box stretching of 1PRS in all parameter regimes, since, as mentioned above, a mild stretching by an additive $n^{\kappa}$ is possible. Concretely, we prove our result via a unitary oracle separation.

\begin{theorem}[Informal]
For every $1<\gamma<\delta$, there exists a unitary oracle $\OO$ relative to which 1PRS with output length $m(n)=\gamma n$ exist, but 1PRS with output length $m(n)=\delta n$ do not exist. This separation holds even when parties have controlled access to $\OO$, its inverse, and its conjugate.
\end{theorem}

As a corollary, our oracle separation implies that one cannot stretch the output length of a 1PRS in a black-box way.

\begin{theorem}[Informal]
There is no fully black-box construction of a 1PRS with output length 
$m(n)=\delta n$ from a 1PRS with output length $m(n)=\gamma n$ for any $1<\gamma<\delta$, 
even when the construction is granted controlled access and access to the inverse and conjugate 
of the shorter 1PRS generator.
\end{theorem}

Our separating oracle is inspired by the Common Haar Random State (CHRS) oracle, which has been used in previous black-box separations~\cite{CCS25,AGL24,BCN25,BMM+25,Goldin-Zhandry}, including separations between 1PRS and multi-copy PRS. The CHRS oracle provides access to a family of states $\{\ket{\psi_i}\}_{i\in\mathbb{N}}$, with one Haar-random state for each length.
We make a key modification: we augment the oracle to support very large quantum circuits acting on an \emph{exponential} (in $i$) number of copies of each state $\ket{\psi_i}$. To construct this augmented oracle and establish the desired separation, we develop several new techniques and technical contributions that may be of independent interest:
\begin{itemize}
\item We give a general technique to upgrade oracle separations of the following kind. Suppose one proves that primitive $A$ exists but primitive $B$ does not relative to an oracle that implements a POVM measurement $\{H, I-H\}$. Our technique can upgrade such a separation generically to one relative to a unitary oracle. The idea is based on Hamiltonian simulation, and the new unitary oracle is $e^{itH}$ for an appropriately chosen $t$. Crucially, two things make the upgrade possible: (i) access to such a unitary allows one to simulate a distinguisher that measures the POVM $\{H, I-H\}$; (ii) in turn, if the POVM measurement is implementable using an auxiliary state $\ket{\Psi}$, then access to $e^{itH}$ can be simulated to inverse-polynomial precision by using polynomially many copies of $\ket{\Psi}$. The latter simulation relies on the quantum Zeno effect.
\item We prove a surprising (at least to the authors) sharpening of the statistical security of the 1PRS construction by Ananth, Gulati, and Lin~\cite{AGL24} in the CHRS model. There, the statistical closeness of the 1PRS to a maximally mixed state was shown to be $O(t^2/2^n)$, when $t$ copies of the common Haar random state are available to the adversary, and $n$ is the length of the seed. Here, we show that this statistical closeness is in fact $O(t/2^n)$. This sharpening is crucial in ruling out black-box stretching of the output length of a 1PRS by any constant factor $c>1$. Without it, our proof would only suffice to rule out $c >2$.
\item We develop a general framework that takes separations in specialized models where adversaries have more oracle power than generators\footnote{This model is known as the two-oracle model in the literature, e.g.\ \cite{HY20,CKM+25,TZ26}.}, and lifts them to separations in the standard oracle model (where generator and adversary have the same power). This formalizes and expands on ideas introduced in a recent work of Chen, Wang, and Xie~\cite{CWX-breakdown}.
\end{itemize}

\section{Technical Overview}\label{sec: tech overview}

\newcommand{\calS}{\mathrm{S}}
\renewcommand{\Gen}{G}
This section is organized as follows. We start by recalling the common Haar random state (CHRS) oracle model, the construction of a 1PRS in this model, and the impossibility of constructing PRS relative to the same oracle. We will then build (and significantly expand) on this idea to define an oracle relative to which 1PRS with a certain stretch exist, but 1PRS with slightly longer stretch do not. We will achieve this step by step, introducing several new ideas along the way.

\paragraph{Common Haar Random State model and quantum pseudorandomness.} The Common Haar Random State (CHRS) model can be viewed as a quantum state generalization of the Common Reference String (CRS) model 
introduced by~\cite{canetti2001universally}. In the CHRS model, introduced in~\cite{CCS25}, we assume a trusted third party, who prepares a family of states $\mathcal{S} = \{\ket{\psi_m}\}_{m \in \mathbb N}$, where $\ket{\psi_m}$ is sampled according to the Haar measure on $m$ qubits $\mu_{2^m}$. The states in the set $\mathcal{S}$ are called CHRS states. All parties in a protocol (including the adversary) have access to polynomially many (in the security parameter $\lambda$) copies of states from $\mathcal{S}$.

Recall that, informally, a 1PRS of input length $n$ and output length $m$ is a $\mathrm{QPT}$ algorithm that outputs a pure state $\ket{\phi_k}$ 
for each $k\in \{0,1\}^n$, such that $$\EE_k\ketbra{\phi_k}{\phi_k}\approx_c\frac{I}{2^m}$$ where $\approx_c$ denotes computational indistinguishability.

Chen, Coladangelo, and Sattath~\cite{CCS25} show that a 1PRS exists relative to the CHRS model, i.e.\ for some $m(n)>n$, there exists a family $\{U_k\}_{k\in \{0,1\}^n}$, such that $\ket{\phi_k} = U_k \ket{\psi_m}$ is a 1PRS. 
This construction was simplified by Ananth, Gulati, and Lin~\cite{AGL24}, who show that the surprisingly simple family given by $U_k=Z^k\otimes I$, for $k \in \{0,1\}^n$, where $Z^k = \bigotimes_{i = 1}^n Z^{k_i}$, already yields:
\begin{equation}\label{eq:AGL24-bound}
    \left\Vert\EE_{\ket{\psi_m}}\left[\EE_k \left(U_k\ketbra{\psi_m}{\psi_m}U_k^\dagger\otimes \ketbra{\psi_m}{\psi_m}^{\otimes(t-1)}\right)-\frac{I}{2^m}\otimes \ketbra{\psi_m}{\psi_m}^{\otimes(t-1)}\right]\right\Vert_1=O\left(\frac{t^2}{2^n}\right).
\end{equation}
Note that the 1PRS construction is simply $\ket{\phi_k} = U_k \ket{\psi_m}$ (i.e.\ it only uses one copy of the CHRS state) and the additional $t$ copies of the CHRS state $\ket{\psi_m}$ should be thought of as given to the adversary. Note also that the guarantee above is \emph{statistical}, rather than computational: security holds against any unbounded adversary, provided it only receives at most $t = \poly(n)$ copies of the CHRS state $\ket{\psi_m}$.

\paragraph{Non-existence of multi-copy PRS in the CHRS model.}

Chen, Coladangelo, and Sattah~\cite{CCS25} also show that, while 1PRS exist in the CHRS model, (multi-copy) PRS do not. We review their attack here, which will be the starting point for our new oracle separation. 

Suppose there is a family of quantum algorithms $\{\Gen_k\}$, such that $\{\Gen_k\ket{\psi_m}\}$ is a PRS family\footnote{ Technically, in full generality, $\Gen_k$ could act on multiple copies of $\ket{\psi_m}$, and even on common Haar states of different sizes, but, for simplicity, we focus on PRS of the form $\{\Gen_k\ket{\psi_m}\}$ in this technical overview, as the general case is similar.}.
To determine whether a challenge state $\ket\phi$ is Haar random or equal to $\Gen_k\ket{\psi_m}$ for some $k$, the first idea that comes to mind, \emph{if one knew $k$}, is to perform a SWAP test between $\ket\phi$ and $\Gen_k\ket{\psi_m}$. If $\ket\phi=\Gen_k\ket{\psi_m}$, this passes the SWAP test with probability $1$. Otherwise, a Haar random $\ket\phi$ passes the SWAP test with probability roughly $1/2$. This gap can be amplified since the adversary is allowed to obtain multiple copies of $\ket\phi$ (since we are considering a multi-copy PRS). By performing, say, $10n$ SWAP tests, the probability of passing all of them in the Haar random case is suppressed to $o(2^{-n})$. 

Crucially, of course, the adversary does not know $k$. However, one can still attack the PRS by leveraging the {\it quantum OR lemma} (Theorem~\ref{thm:OR-test}), which says that there exists an algorithm (which we refer to as the {\it OR tester}) that, given a series of tests (that is, projectors) $\{\Pi_k\}_{k=1}^N$ and a state $\ket\phi$, distinguishes between the following two cases.
\begin{itemize}
    \item (``YES" case) There exists $k$ such that $\ket\phi$ passes the test $\Pi_k$ with probability $\Omega(1)$. In this case, the OR tester will accept with probability $\Omega(1)$.
    \item (``NO" case) For every $k$, $\ket\phi$ passes the test $\Pi_k$ with probability $o(1/N)$. In this case, the OR tester will accept with probability $o(1)$.
\end{itemize}

The adversary runs the OR tester with $\Pi_k$ corresponding to the test that creates $10n$ copies of $\Gen_k\ket{\psi_m}$, and performs $10n$ SWAP tests with $10n$ copies of $\ket\phi$. In our case, $N = 2^n$, and so $10n$ SWAP tests are enough to suppress the acceptance probability in the NO case to $o(2^{-n})$ for each $k$.

\paragraph{Separating 1PRS with different stretch: the first attempt.} Now, we turn our attention to separating 1PRS with different stretch. The first roadblock is that we can no longer run an OR tester as above, since the latter required $10n$ copies of the challenge state $\ket{\phi}$: a multi-copy PRS adversary can obtain that many copies, but a 1PRS adversary can only obtain one. Recall that the $10n$ copies were required because of an inherent limitation of the SWAP test, which always accepts with probability at least $1/2$.

Our goal is to modify the oracle in such a way that it still allows us to construct a secure 1PRS of small stretch (say $1.1n$ output length), but gives us enough power to break a longer stretch 1PRS. The question we need to address then is: what concrete difference between short stretch and long stretch 1PRS can we leverage? 

Our first key observation is the following. Recall that in the 1PRS construction from Ananth, Gulati, and Lin~\cite{AGL24}, the (statistical) distinguishing advantage is $O(t^2/2^n)$, where $t$ is the number of copies of the common Haar state $\ket{\psi_m}$, and $n$ is the key-length. Crucially, the security is a function of the key length $n$, and not of the output length~$m$! Thus, \emph{regardless of the stretch}, the \cite{AGL24} 1PRS is secure as long as an (unbounded) adversary has $o(2^{n/2})$ copies of the common Haar state. This suggests the following idea: consider an adversary with ``access'' to, say, $2^{m/3}$ copies of the common Haar state of length $m$. This number of copies is not enough to break an \cite{AGL24} 1PRS with stretch $m(n) = 1.1 n$, since $2^{m/3} = 2^{1.1n/3} =o( 2^{n/2})$. However, a 1PRS with stretch, say, $m(n) = 10n$ is no longer guaranteed to be secure, since $2^{m/3} \gg 2^{n/2}$! Now, is there actually an attack on any 1PRS with stretch $m(n) = 10n$ in this model? Yes, one can build an attack based on a generalized version of the SWAP test, namely a ``permutation test"~\cite{KNY08-perm}, which checks overlap with the symmetric subspace on more than two registers. The permutation test satisfies the following, for any pair of states $\ket{\phi}$ and $\ket{\xi}$ of the same length $i$, and any integer $r\geq 2$:
\begin{equation}\label{eq:GSWAP}\Pr[\text{accept}] = \tr\left(\Pi_{\text{sym}}^{(r,2^i)}\cdot \ketbra{\phi}{\phi}\otimes \ketbra{\xi}{\xi}^{\otimes (r-1)}\right)=\frac{1}{r}+\frac{r-1}{r}|\braket{\phi|\xi}|^2\,,\end{equation}
where $\Pi_{\text{sym}}^{(r,2^i)}$ is the projection onto the symmetric subspace of $r$ registers of dimension $2^i$. Notice that the ``soundness'' of the test is controlled by the number of copies $r$ (improving inverse-linearly with $r$) and the fidelity between $\ket{\phi}$ and $\ket{\xi}$. Crucially, the latter, on average over a Haar random state $\ket{\phi}$, scales inverse-exponentially with the number of qubits of the state (which in our case is the output length $m$ of the 1PRS). 

Let us now concretely see how the permutation test can be used to break a 1PRS with sufficiently long stretch, e.g.\ input length $n$ and output length $m= 10n$. Here, for now, we again assume that \emph{the candidate 1PRS generator $\Gen_k$ only uses a single copy of $\ket{\psi_m}$, and no other CHRS states, and outputs a state of the same length $m$} (in general, the 1PRS is allowed to use many copies of the CHRS states---we will see later that this is a crucial obstacle that we will have to overcome). Denote $\ket{\phi_k} = \Gen_k \ket{\psi_m}$. Given, say, $r = 2^{m/3}$ copies of $\ket{\psi_m}$, an unbounded adversary is able to generate $2^{m/3}$ copies of $\ket{\phi_k}$ for an \emph{arbitrary} $k$. Then, consider running an OR tester, where each test $\Pi_k$, for $k \in \{0,1\}^n$, corresponds to one permutation test between the challenge state $\ket{\phi}$ and $2^{m/3}$ copies of $\ket{\phi_k}$. There are two cases: 
\begin{itemize}
\item If there exists $k$ such that $\ket{\phi} =\ket{\phi_k}$, then $\ket{\phi}$ passes the corresponding $\Pi_k$ permutation test with probability 1. So, the OR tester accepts with probability $\Omega (1)$.
\item If $\ket{\phi}$ is sampled from the Haar measure on $m$ qubits, then, with overwhelming probability, it holds for all $k$ that $|\langle \phi|\phi_k\rangle|^2 = O(2^{-m})$. This means that each permutation test accepts with probability 
$$ \frac{1}{r}+\frac{r-1}{r}|\braket{\phi|\phi_k}|^2  = 2^{-m/3} + O(2^{-m}) = o(2^{-n})\,,$$
since $m = 10n$.
So, the OR tester accepts with probability $o(1)$.
\end{itemize}

So, the OR tester is able to break the 1PRS candidate. This discussion suggests that, if we wish to separate 1PRS with different stretch, we should ``augment'' the CHRS oracle so that it does not only provide \emph{one} copy of a common Haar state $\ket{\psi_m}$ per query, but it additionally gives the ability to run an OR tester algorithm that uses \emph{exponentially many copies}\footnote{Here, ``exponentially many" means an exponential in the number of qubits $m$, not in the security parameter.} (for an appropriately chosen exponential) of the length-$m$ state $\ket{\psi_m}$ for each $m$. 

Now, let us sharpen the parameter choice in the example above. Suppose the augmented oracle provides the adversary $r = 2^{m/L}$ copies of $\ket{\psi_m}$ for some appropriately chosen $L$. Then, following the above argument, if we wish to separate 1PRS with output lengths $d\cdot n$ and $D \cdot n$, we should choose $L$ so that: 
$$ Dn/L > n\,,$$
and 
\begin{equation}
\label{eq:dnL}
dn/L < n/2\,.
\end{equation}
It is clear that such an $L$ exists if and only if $D>2d$. If the latter is true, we can simply pick $L = D$. So, \emph{if} the above outline worked as intended without any issues, it would yield a separation between 1PRS with output lengths $m(n)$ and $C \cdot m(n)$ for any $C>2$, and any linear function $m(n)$.

\paragraph{Where the first attempt fails.}
There are two important issues that we have not considered.

\par\noindent\textbf{Issue 1.}
Whenever we give more power to the adversary, by considering a stronger oracle, a 1PRS generator also gains the same additional power. As a result, if the adversary is allowed to use $2^{m/L}$ copies of $\ket{\psi_m}$, for some constant $L>1$, thanks to the new oracle, the \emph{generator} can also use the same number of copies. As a result, generating $2^{m/L}$ copies of $\ket{\phi_k}$ might itself require $2^{2m/L}$ copies of $\ket{\psi_m}$, exceeding the supply available to the adversary. In other words, if we strengthen the oracle to provide that many copies, then the generator can also use that many, and so on: this appears to be an endless arms race!

\par\noindent\textbf{Issue 2.}
The generator is not limited to using $\ket{\psi_m}$, but can also use CHRS states of other lengths. In particular, it may output an $m$-qubit state using only a short CHRS state, say $\ket{\psi_3}$. An oracle that provides $2^{i/L}$ copies of the length-$i$ state then gives at most $8$ copies of $\ket{\psi_3}$. A permutation test using these many copies has a constant acceptance probability even on an orthogonal challenge state, which is far too large to allow for an OR test over $2^n$ keys.\\

\paragraph{Resolving Issue 1: An asymmetric oracle model.}
The first insight is that it is actually fruitful to consider an asymmetric oracle model where generator and adversaries have access to different resources. More precisely, this is a model where the oracle consists of two parts $\OO = (\OO_1, \OO_2)$. The generator can only access $\OO_1$, while the adversary can access both $\OO_1$ and $\OO_2$. Such a model has already been considered classically, and referred to as the ``two-oracle'' model. In fact, it is already known in the classical setting~\cite{HR04} that an oracle separation between two primitives $A$ and $B$ in this model (where $A$ exists and $B$ does not) is sufficient to imply that there is no fully black-box construction of $B$ from $A$. This result also extends to the quantum setting and has found its application in various proofs of separation, including~\cite{HY20,CKM+25,TZ26}. We will refer to the two-oracle model simply as the asymmetric model going forward.

Fortunately, ignoring issue 2, which we will tacke in a moment, we already have an outline for how to separate 1PRS with output lengths $m(n)$ and $C \cdot m(n)$ for any $C>2$, and any linear function $m(n)$, in an asymmetric model where the adversary additionally gets $2^{m/L}$ copies of the CHRS state $\ket{\psi_m}$, for some appropriate constant $L>1$. 


In the rest of the technical overview, we will thus focus on obtaining an oracle separation in this asymmetric model. At the end of the overview, we will outline how one can lift this to a standard oracle separation, by developing a general framework, which formalize ideas from~\cite{CWX-breakdown}. 

\vspace{5mm}

\paragraph{Resolving issue 2: Reversing the generation procedure.}
We have described a mechanism that allows the adversary to use more copies of the CHRS states than the generator, but is this enough to carry out an attack?

As noted above, the difficulty with shorter CHRS states is that the oracle does not provide enough copies for our proposed attack. The saving grace is that a secure generator cannot rely too heavily on these shorter states. We explain this intuition as follows.

Suppose the longer 1PRS generator uses copies of a single constant-length CHRS state, say $\ket{\psi_3}$, and no other CHRS states. Let us first consider the case where it uses only one copy. Why can this not be a secure 1PRS? For each fixed seed $k$, even without the knowledge of the CHRS state, the output lies in a subspace of dimension at most $2^3=8$. Crucially, membership in this subspace can be tested efficiently by reversing the generation procedure (given $k$), which is \emph{independent} of the particular CHRS state. The challenge Haar-random state on $m$ qubits passes this subspace test with probability $\approx 2^{-m+3}$, which is small enough to instantiate an OR test.

One might think that using multiple copies of the CHRS state perhaps could help the generator. However, because all copies are identical, their joint state lies in the symmetric subspace of the corresponding Hilbert space. Crucially, when the dimension of each copy is constant, the dimension of this symmetric subspace grows only \emph{polynomially} in the number of copies. In comparison, the challenge register has \emph{exponential} dimension, so an independent Haar-random state still passes the corresponding subspace test with exponentially small probability. Thus, even polynomially many copies of a constant-length CHRS state do not help the security of generators of this form.

This observation suggests a new approach: if the candidate generator relies heavily on constant-length CHRS states, we can ``give up" identifying these states precisely, but instead check whether the challenge lies in the subspace determined by the symmetry of the copies of the CHRS states, and the fact that \emph{ancillas} are initialized to $\ket{0}$.
The natural way to perform this check is to \emph{reverse} the generation procedure, starting from the challenge state, and test the recovered input registers.

\paragraph{The test in more detail, and why it works.}
We capture the intuition above with the following test. Fix a seed $k$, and let $\ket{\tilde{\phi}_k}$ denote the garbage state of the work registers that is discarded by the generator. We can prepare this state by running $G_k$ on fresh CHRS copies and setting aside its output. Given a challenge $\ket{\phi}$, we then apply the reversed generation procedure on $\ket{\phi}\otimes\ket{\tilde{\phi}_k}$.
If $\ket{\phi}$ is the state generated by $G_k$, this recovers the original input registers: the ancillas initialized to $\ket{0}$, and the copies of the CHRS states. Our test then checks whether this is indeed the case. More precisely, it accepts if all the following checks pass:
\begin{enumerate}
    \item For the ancilla registers, we measure in the computational basis and check that all outcomes are zero.
    \item For the CHRS registers, we perform a generalized SWAP test for each CHRS state length $i$. The test acts on the recovered registers that should contain polynomially many copies of $\ket{\psi_i}$, together with $2^{i/L}$ reference copies supplied by the oracle, where $L$ is an appropriately chosen constant such that $m/L>n$.
\end{enumerate}

The ancilla checks have perfect soundness; however, it is not obvious how accurate the generalized SWAP test is on the CHRS registers.
To show the test indeed achieves the required precision for the OR test, the intuition is the following.
    \begin{itemize}
        \item For large $i$, the generalized SWAP test is carried out with exponentially many additional copies of the corresponding CHRS states $\ket{\psi_i}$, and this results in a test, with very high soundness, that checks whether the tested registers are in the state $\ket{\psi_i}$, following a similar calculation as described earlier.
        \item For small $i$, the test reduces to checking whether the tested registers lie in the symmetric subspace, but otherwise morally ``gives up'' on identifying the precise states they contain. As discussed above, the small overall dimension of this symmetric subspace makes this weaker check still sufficient, in conjunction with the other tests.
    \end{itemize}
More precisely, using standard facts about the spectrum of the permutation test, we show that the average of the largest $2^m$ eigenvalues of the test above (which is the best subspace that an $m$-qubit Haar random challenge state can hope to produce after the reversed generation procedure) is at most $2^{-(1+o(1))m/L}=o(2^{-n})$.
Thus, even the $2^m$ most accepting directions have negligible average acceptance. This bounds the probability that an independent Haar-random challenge passes the test by $o(2^{-n})$, which is small enough to run an OR test.

\paragraph{The two remaining desiderata.}
At this point, we have outlined how to obtain an oracle separation between $n \rightarrow m(n)$ 1PRS and $n \rightarrow C\cdot m(n)$ 1PRS, where $m(n)$ is any linear function of $n$, and $C>2$. We have argued this in the asymmetric oracle model, but as we will describe at the end of the overview, this separation can be lifted to the standard oracle model. Thus, there are two major remaining desiderata:
\begin{itemize}
\item Ruling out \emph{any} linear stretch, i.e.\ obtaining a separation between $n \rightarrow m(n)$ 1PRS and $n \rightarrow c \cdot m(n)$, for any $c>1$.
\item Upgrading our separating oracle to a \emph{unitary} oracle. Currently, the separating oracle can be thought of as a \emph{channel} oracle: it consists of the CHRS oracle, plus an oracle that lets one implement an OR tester POVM measurement.
\end{itemize}
We discuss next how to achieve each of these desiderata.

\paragraph{Ruling out \emph{any} linear stretch.} 
The reason why our analysis so far has only been able to rule out stretching by a factor $C>2$ is that the statistical guarantee of the short 1PRS construction from \cite{AGL24}, from Equation~\eqref{eq:AGL24-bound}, is $O(t^2/2^n)$. In particular, the exponent in the numerator is what imposed the constraint \eqref{eq:dnL}. For instance, improving the bound to $O(t/2^n)$ would directly improve our separation, ruling out stretching by any factor $C>1$. So, is such a bound just wishful thinking?




Perhaps surprisingly, we show that this stronger bound in fact holds:
\begin{equation}\label{eq:informal-tight-bound}
    \left\Vert\EE_{\ket{\psi_m}}\left[\EE_k \left((Z^k\otimes I)\ketbra{\psi_m}{\psi_m}(Z^k\otimes I)-\frac{I}{2^m}\right)\otimes \ketbra{\psi_m}{\psi_m}^{\otimes t}\right]\right\Vert_1=O\left(\frac{t}{2^n}\right).
\end{equation}

In order to prove the stronger security guarantee, we will have to depart from the type-vector analysis of~\cite{AGL24}. The reason is that avoiding all the collision pairs among the $t$ registers immediately yields a term $O(t^2/2^n)$. 

First, recall that applying a random $Z$ gate on the first $n$ qubits of $\ket{\psi_m}$ is equivalent to measuring the first $n$ qubits in the computational basis. So, for every $k\in \{0,1\}^n$, let $\Pi_k=\ketbra{k}{k}\otimes I$, then it suffices to prove that
$$\left\Vert\EE_{\ket{\psi_m}}\left[ \left(\sum_k\Pi_k\ketbra{\psi_m}{\psi_m}\Pi_k-\frac{I}{2^m}\right)\otimes \ketbra{\psi_m}{\psi_m}^{\otimes t}\right]\right\Vert_1=O\left(\frac{t}{2^n}\right).$$
Let
$$\rho=\EE_{\ket{\psi_m}}\left[ \left(\sum_k\Pi_k\ketbra{\psi_m}{\psi_m}\Pi_k\right)\otimes \ketbra{\psi_m}{\psi_m}^{\otimes t}\right]\quad\text{and}\quad
\sigma=\frac{I}{2^m}\otimes \EE_{\ket{\psi_m}}\left[  \ketbra{\psi_m}{\psi_m}^{\otimes t}\right]\,.$$

Notice that the first register of $\sigma$ is completely independent of the state $\ket{\psi_m}$. Hence, our strategy for bounding $\|\rho-\sigma\|_1$ is to find a projector $\Pi$ that reduces the correlations of the first register of $\rho$ from the rest, so then we can hope to achieve the desired bound on $\|\Pi(\rho-\sigma)\Pi\|_1$. At the same time, we would like $\Pi$ to have a large overlap with both $\rho$ and $\sigma$, so then we can bound both $\|\rho-\Pi\rho\Pi\|_1$ and $\|\sigma-\Pi\sigma\Pi\|_1$. Then we will bound $\|\rho-\sigma\|_1$ by 
$$\|\rho-\sigma\|_1\leq \|\rho-\Pi\rho\Pi\|_1+\|\Pi(\rho-\sigma)\Pi\|_1+\|\sigma-\Pi\sigma\Pi\|_1\,.$$

Our main observation is that since $\ket{\psi_m}$ is sampled from the real Haar measure, we can write $\ket{\psi_m}=\ket{\psi_{m,k}}+\ket{\hat \psi_{m,k}}$, where $\ket{\psi_{m,k}}=\Pi_k\ket{\psi_m}$ and $\ket{\hat \psi_{m,k}}=(I-\Pi_k)\ket{\psi_m}$, such that $\|\ket{\psi_{m,k}}\|^2$ and the directions of $\ket{\psi_{m,k}}$, $\ket{\hat \psi_{m,k}}$ are independent. In other words, we break down the sampling of $\ket{\psi_m}$ from the unit sphere as the process of first selecting the directions of $\ket{\psi_m}$ in the space that $\Pi_k$ projects to and its orthogonal complement, and then selecting the weight of the state in each one of the two spaces. 

From this point of view, notice that for every $k$, the first register of $\Pi_k\ket{\psi_m}\otimes \ket{\psi_m}^{\otimes t}=\ket{\psi_{m,k}}\otimes \ket{\psi_m}^{\otimes t}$ is already independent of the direction of $\ket{\hat\psi_{m,k}}$. Hence, after applying $(I-\Pi_k)$ on each one of the other $t$ registers, the only correlation left between the first register and the rest would be the weight $\|\ket{\psi_{m,k}}\|^2$. 

With this in mind, we consider the projector
$$\Pi=\sum_k\Pi_k\otimes (I-\Pi_k)^{\otimes t}\,.$$
Then $\|\Pi(\rho-\sigma)\Pi\|_1$ is bounded by
$$ \sum_{k}\left\Vert\EE_{\ket{\psi_m}}\left[ \left(x_k-\frac{1}{2^n}\right)\ketbra{\hat{\psi}_{m,k}}{\hat{\psi}_{m,k}}^{\otimes t}\right]\right\Vert_1\,.$$
where $x_k=\|\ket{\psi_{m,k}}\|^2$. We notice that replacing $\ket{\hat{\psi}_{m,k}}$ with $\ket{\hat{\psi}_{m,k}}/\|\ket{\hat{\psi}_{m,k}}\|$ would make the above term to vanish, since the direction of $\ket{\hat{\psi}_{m,k}}$ is independent of $x_k$, and $\EE [x_k]=1/2^n$. Moreover, $\ket{\hat{\psi}_{m,k}}/\|\ket{\hat{\psi}_{m,k}}\|$ is $x_k$-far from $\ket{\hat{\psi}_{m,k}}$, and thus, $t$ copies of $\ket{\hat{\psi}_{m,k}}/\|\ket{\hat{\psi}_{m,k}}\|$ are $2tx_k$-far from $\ket{\hat{\psi}_{m,k}}^{\otimes t}$ in trace distance. Hence, $\|\Pi(\rho-\sigma)\Pi\|_1$ is at most
$$ \sum_{k}\EE_{\ket{\psi_m}}\left[ 2tx_k\left(x_k+\frac{1}{2^n}\right)\right]\,,$$
which is bounded by $O(t/2^n)$ since $\sum_k x_k=1$ and $\EE[x_k^2]$ is roughly $1/2^{2n}$.

Through direct calculations, we can show that $\Tr((I-\Pi)\rho)$ and $\Tr((I-\Pi)\sigma)$ are both bounded by $O(t/2^n)$. Applying the gentle measurement lemma yields that $\|\rho-\Pi\rho\Pi\|_1$ and $\|\sigma-\Pi\sigma\Pi\|_1$ are both bounded by $O(\sqrt{t/2^n})$, implying an overall bound of $O(\sqrt{t/2^n})$ for $\|\rho-\sigma\|_1$.

Notice that this weaker bound suffices to generate the desired qualitative conclusion with respect to the security of the 1PRS construction: the scheme is secure as long as the adversary has access to at most $2^{(1-\epsilon)n}$ copies of the state $\ket{\psi_m}$, for any $\epsilon>0$. Nonetheless, with a careful analysis of the symmetries of the Haar measure, we derive the stronger $O(t/2^n)$ bound for the terms $\|\rho-\Pi\rho\Pi\|_1$ and $\|\sigma-\Pi\sigma\Pi\|_1$, and consequently for $\|\rho-\sigma\|_1$ as well.

\paragraph{Lifting to a unitary oracle.}
We have seen how access to exponentially many copies of each CHRS state can help separate shorter and longer 1PRS. A polynomial-time adversary, however, cannot explicitly store and process all these copies. In this paragraph, we seek a unitary oracle that gives the adversary the distinguishing power of the OR test, without requiring it to maintain the test's exponentially large workspace.

In fact, the lifting technique for the CHRS state oracle has been well-founded, e.g.\ from~\cite{Goldin-Zhandry}. But what about the other oracle that supplies exponentially many CHRS copies and the OR test over them?

The starting point is to view the OR test as a \emph{measurement} (or more abstractly, a \emph{channel}) on the input register. At a high level, the test proceeds as follows:
\begin{enumerate}[label=\alph*)]
    \item Acquire the CHRS states and the ancillas needed for the OR test.
    \item Apply a unitary circuit to the input register, the CHRS registers, and the ancillas.
    \item Measure the first qubit in the computational basis, accepting if the outcome is $1$.
\end{enumerate}

Thinking of the CHRS states as fixed, this procedure induces a \emph{two-outcome POVM measurement} on the input register, which depends only on the algorithm itself and the CHRS states. As usual, we describe this measurement by an Hermitian operator operator $0\preceq H\preceq I$. Once the CHRS states are chosen and fixed, $H$ is fixed as well, and this is the only measurement which we need to provide access to.

However, it is not obvious how to pass from a general POVM to a unitary, especially when implementing the measurement requires a very large workspace. How can we provide access to this measurement without requiring such a large workspace?

\emph{Exponentiation} provides a natural way to convert $H$ into a unitary. Specifically, the operator
$e^{itH}$
is unitary for every real $t$ and acts on the same register as $H$.

The next step is to figure out how to instantiate the measurement $H$ by the unitary $e^{itH}$.
For an appropriately chosen $t$, the following circuit, known as the Hadamard test, captures the acceptance probability $\braket{\phi|H|\phi}$ sufficiently well for our purposes:\footnote{Here, to avoid confusion, we use the sans-serif $\mathsf{H}$ to denote the single-qubit Hadamard gate. Also, throughout the formal proof we choose $t=\pi$, but any fixed $t\in (0,2\pi)$ suffices.}
\begin{center}
\begin{quantikz}
\lstick{$\ket{0}$} & \gate{\mathsf{H}} & \ctrl{1} & \gate{\mathsf{H}} & \meter{} \\
\lstick{$\ket{\phi}$} & & \gate{e^{itH}} & &
\end{quantikz}
\end{center}

Thus, the unitary oracle provides $U=e^{itH}$ (with controlled access), allowing the adversary to run the circuit above and recover the distinguishing power needed from the measurement $H$. Recall that $H$ encodes the OR test described earlier, which breaks the security of the longer 1PRS. We therefore obtain an attack on the longer 1PRS in the asymmetric model.

\emph{But wait...} Could access to $e^{itH}$ give the adversary so much additional power that it also breaks the shorter 1PRS?

\paragraph{How the shorter 1PRS still survives.}
Recall our earlier statistical bound: the shorter 1PRS remains secure even against adversaries given up to a certain exponential number of \emph{copies} of the CHRS state. To show security against an adversary with access to $e^{itH}$, we show that such access can be simulated sufficiently well given access to the measurement $H$ (which in turn can be implemented with an appropriately bounded number of copies of the CHRS states). 

The key is to simulate $e^{itH}$ while keeping the number of CHRS copies under control. To see how, we examine the structure of $H$ more closely. Let $\ket{\Psi}$ denote the joint initial state of the CHRS registers and the ancillas, and let $C$ denote the unitary circuit used by the measurement procedure. Then
\[
    H=(I\otimes\bra{\Psi})C^\dagger
    (\ketbra{1}{1}\otimes I)
    C(I\otimes\ket{\Psi}).
\]
A natural first attempt is to exponentiate each term within this expression. One might therefore hope that
\begin{equation}\label{eq:question_effective_ham}
    e^{itH}\stackrel{?}{=}
    (I\otimes\bra{\Psi})C^\dagger
    (e^{it\ketbra{1}{1}}\otimes I)
    C(I\otimes\ket{\Psi}).
\end{equation}
However, this does not hold in general. To understand why, it is helpful to adopt a Hamiltonian viewpoint.

Let $V_0$ be the full Hilbert space on which $C$ acts, and let $V\subseteq V_0$ be the subspace in which the CHRS and ancilla registers are fixed to $\ket{\Psi}$. Define
\[
    H_0=C^\dagger(\ketbra{1}{1}\otimes I)C\,\quad\text{and}\quad H=H_0|_V=(I\otimes\bra{\Psi})H_0(I\otimes\ket{\Psi})\,.
\]
Under this identification, $H$ is the \emph{effective Hamiltonian} on $V$, and $e^{itH}$ describes its time evolution.
The time evolution of $H_0$ has the simple form
\[
    e^{itH_0}
    =C^\dagger(e^{it\ketbra{1}{1}}\otimes I)C\,.
\]
This can be implemented directly in the plain model, without access to the CHRS states. The obstacle is that restriction and exponentiation do not commute, i.e.\ \eqref{eq:question_effective_ham} is not always true in general:
\[
   e^{itH}= e^{it(H_0|_V)}\neq e^{itH_0}\big|_V=(I\otimes\bra{\Psi})e^{itH_0}(I\otimes\ket{\Psi})\,.
\]
Indeed, evolution under $H_0$ need not preserve $V$: a state initially in $V$ may develop a component outside $V$, and its subsequent evolution can depend on parts of $H_0$ not captured by the effective Hamiltonian $H$. Evolution under $H$, however, remains confined to $V$ (since by definition $H$ is the effective Hamiltonian on $V$).
Thus, we need to find a method to \emph{confine} the time evolution of $H_0$ inside the subspace $V$. This is where we use the quantum Zeno effect, or in a more eye-catching way: \emph{A watched quantum pot never boils.} More rigorously, we need the following.

\begin{lemma}[Informal]\label{lem:informal-ham-sim}
    Let $\eta$ be a positive integer. When $\eta\to\infty$:
    $$\left\|(e^{itH_0/\eta}\big|_V)^\eta- e^{itH_0|_V}\right\|=O\left(\frac{1}{\eta}\right)\,.$$
\end{lemma}

Intuitively, if we ``watch" the state frequently, the state will be ``freezed" in the subspace $V$, and we are able to get a good simulation of the effective Hamiltonian on $V$. Choosing $\eta$ to be a sufficiently large polynomial therefore yields an inverse-polynomial simulation error, but increasing the total number of copies needed only by a polynomial factor.

\paragraph{Lifting the asymmetric model to a standard oracle separation.}
We have now established a separation in the asymmetric model. Although this already suffices to rule out fully black-box constructions, recent literature favors separations in the standard oracle model, where the generator and the adversary have the same oracle access power. A further contribution of this paper is a technique for lifting our separation to this setting, generalizing the diagonalization technique in~\cite{CWX-breakdown}.

The first observation is that, to violate the usual definition of the security, it suffices for a nonuniform adversary to achieve an inverse-polynomial advantage \emph{infinitely often}. Thus, we may enumerate all generators $M_1,M_2,\cdots$, and attack each of them by a fixed schedule: at stage $s$, we attack $M_{\kappa(s)}$ on a suitably chosen security parameter $\Lambda_s$. We choose $\kappa$ to visit every positive integer infinitely often, thus breaking all possible generators.

We construct the oracle as follows. Suppose the asymmetric model provides a pair of oracles $(A,B)$, where generators access only $A$, while adversaries access both $A$ and $B$. Our standard separating oracle will be $\mathcal O=(A,\widehat B)$, where $\widehat B$ is a padded version of $B$.

Fix a stage $s$, and suppose the generator we want to break runs in time $n^\tau$ for each security parameter $n$. We place the oracle blocks introduced at this stage at query lengths greater than $\Lambda_s^{2\tau}$. On security parameter $\Lambda_s$, the generator of size at most $\Lambda_s^\tau$ cannot access these blocks. However, an adversary with a larger polynomial size bound, say $n^{4\tau}$, can access the padded $B$-oracle needed for its attack. This recreates the imbalance, despite both parties having access to the same oracle.

What about the blocks introduced at earlier stages? Here our treatment differs from~\cite{CWX-breakdown}. Roughly speaking, we space the stages far enough apart, so that every padded $B$-oracle blocks introduced in earlier stages are logarithmic in the current security parameter. These blocks can therefore be learned (e.g.\ through tomography) in polynomial time and then hardwired into the generator.\footnote{In the formal proof, we use nonuniformity and the Solovay--Kitaev theorem to compile the smaller $\widehat B$-oracle blocks directly into the circuit, which achieves negligible approximation error.} This eliminates the generator's access to $\widehat B$, at the cost of its nonuniformity. Thus, we require that in the asymmetric model, there does not exist even a \emph{nonuniform} version of the primitive we want to break. This is reflected in the following informal theorem.

\begin{theorem}[Informal]
Suppose $P$ has an efficient $A$-oracle implementation secure against nonuniform $(A,B)$-oracle adversaries, whereas no efficient \emph{nonuniform} $A$-oracle implementation of $Q$ is secure against such adversaries.
Then there exists a standard oracle $\mathcal O$ relative to which $P$ exists but $Q$ does not.
\end{theorem}


\paragraph{Organization of the proof.}
Our paper is organized as follows.
In Section~\ref{sec: oracle separation}, we present and prove the main theorem, the separation of short 1PRS and long 1PRS. Section~\ref{sec: oracle separation} consists of a few parts.  Section~\ref{sec:hamiltonian-simulation} presents and proves the formal version of the Hamiltonian simulation theorem (Lemma~\ref{lem:informal-ham-sim}). Section~\ref{sec:short-1PRS} proves the existence of the short 1PRS under our oracle, assuming the new tight bound on the security of the \cite{AGL24} 1PRS construction, which we carry out in Section~\ref{sec:short-1PRS-security}. Section~\ref{sec:longer-1PRS} proves the non-existence of longer 1PRS. Section~\ref{sec: lifting theorem} presents a general framework to lift asymmetric oracle separations to standard ones. Note that, while there is a route to the proof in Section~\ref{sec: oracle separation} based on the lifting framework presented in Section~\ref{sec: lifting theorem}, the actual proof in Section~\ref{sec: oracle separation} is slightly more direct, and is technically independent of Section~\ref{sec: lifting theorem}. Nonetheless, we include Section~\ref{sec: lifting theorem} as it may be of independent interest.


\section{Preliminaries}
We use $S_n$ to denote the set of permutations on $[n]$. Let $\mathbb S^d$ be the $d$-dimensional sphere. Denote the Haar measure on the set of states on the $d$-dimensional Hilbert space $\mathbb C^d$ as $\mu_d$, and the \emph{real} Haar measure on the set of real states on $\mathbb R^d\subset \mathbb C^d$ as $\mu_{d}^{\mathbb R}$. For any finite set $K$, we write $k \gets K$ to mean that $k$ is sampled uniformly at random from $K$. For an operator $H$, we use the notation $\|H\|_1$ to denote its trace norm and $\|H\|_\mathrm{op}$ to denote its operator norm, and let $\|H\|_\text{F}:=\sqrt{\Tr(H^\dagger H)}$ be its Frobenius norm. For a matrix $H$, denote $H^T$ as the transpose of $H$, and $\overline{H}$ as the (element-wise) complex conjugate of $H$. We also use $\|u-v\|$ to denote the Euclidean distance between two vectors $u,v$. Let $L\in \NN$, then a function $f:\CC^{n}\to \mathbb{R}$ is an $L$-Lipschitz function with respect to the Euclidean distance, if $|f(\ket{\psi})-f(\ket{\phi})|\leq L\|\ket{\psi}-\ket{\phi}\|$, for all states $\ket{\psi},\ket{\phi}\in \CC^n$. We use the notation $A^{(\cdot)}$ to refer to an algorithm (classical or quantum) that makes queries to an oracle. For a Hilbert space $\mathcal H$, we denote the $k$-th symmetric tensor space of $\mathcal H$ by $\Sym_{k}\mathcal H$. We will use $\Pi_{\text{sym}}^{(n,d)}$ to refer to the projection from the $n$-th tensor product space $(\mathbb C^d)^{\otimes n}$ onto the $n$-th symmetric tensor subspace $\Sym_{n}\mathbb C^d$. 

\subsection{Pseudorandom States}

\begin{definition}[Pseudorandom States (PRS), adapted from~\cite{JLS}]\label{def:PRS}
    A pseudorandom state generator with output length $m(\cdot)$ is a QPT algorithm $G$ that, on input $(1^n,k)$ with $k\in\{0,1\}^n$, outputs a \textit{pure} state $\ket{\phi_k}$ consisting of $m(n)$ qubits, such that, for any polynomial $t=t(n)$ and any QPT adversary $\mathcal A$ that may be nonuniform (in the sense that it takes a classical advice of polynomial length in $n$), there exists a negligible function $\mathsf{negl}$ such that for every $n$ and every advice $c$,
    \begin{equation}\label{eq:def of PRS security}
        \left|\Pr_{k\in \{0,1\}^n}\left[\mathcal A(\ket{\phi_k}^{\otimes t(n)};c)=1\right]-\Pr_{\ket{\phi}\sim \mu_{2^m}}\left[\mathcal A(\ket{\phi}^{\otimes t(n)};c)=1\right]\right|=\mathsf{negl}(n),
    \end{equation}
    where $\mu_{2^m}$ is the Haar measure on $m=m(n)$ qubits.
\end{definition}

\begin{definition}[Single-copy Pseudorandom States (1PRS), adapted from~\cite{MY22a}] A single-copy
pseudorandom state generator (1PRS) with output length $m(\cdot)$ is a PRS (as in Definition~\ref{def:PRS}) with $m(n)>n$, for every $n\in \NN$, except that the security definition in Equation~\ref{eq:def of PRS security} only holds for $t=1$.
\end{definition}

Notice that the \emph{stretch} requirement in the definition of 1PRS, i.e.\ $m(n)>n$ for every $n$, is essential to rule out trivial constructions. For instance, one can easily verify that by simply setting $\ket{\phi_k}=\ket{k}$, the uniform mixture of the states $\{\ket{\phi_k}\}_{k\in\{0,1\}^m}$ is computationally indistinguishable from the maximally mixed state on $n$ qubits. 

\subsection{Useful lemmas on matrix norm}
The following lemma says that, if post-selection on an ancilla register succeeds with high probability and gives an operation close to \(U\), then simply tracing out the ancillas also gives an operation close to \(U\).
\begin{lemma}\label{lem:trace-out-vs-measurement}
    Let $\mathcal H_1,\mathcal H_2$ be Hilbert spaces, and $V:\mathcal H_1\to \mathcal H_1\otimes \mathcal H_2$ be an isometry. Let $\ket{\psi}\in \mathcal H_2$ be an arbitrary state. Suppose $U$ is a unitary on $\mathcal H_1$. Let $\Tr_2$ be the channel defined as tracing out the $\mathcal H_2$ register. Assume
    $$\left\Vert(I\otimes\bra{\psi})V-U\right\Vert_{\text{op}}\leq \epsilon\,.$$
    Then
    $$\left\Vert \Tr_2\circ V-U\right\Vert_\Diamond\leq 4\epsilon\,,$$
    where $V,U$ also denote the corresponding channel by abuse of notation.
\end{lemma}

\begin{proof}
\newcommand{\op}{\text{op}}
Let
$$
A=(I\otimes\bra{\psi})V.
$$
Then \(\|A-U\|_{\op}\leq\epsilon\), and \(\|A\|_{\op}\leq1\). Write
$$
V=A\otimes\ket{\psi}+W,
$$
where the second-register component of \(W\) is orthogonal to \(\ket{\psi}\). Hence
$$
(\Tr_2\circ V)(X)=AXA^\dagger+\Tr_2(WXW^\dagger).
$$

Since \(V^\dagger V=I\),
$$
W^\dagger W=I-A^\dagger A=U^\dagger U-A^\dagger A=U^\dagger(U-A)+(U^\dagger-A^\dagger)A.
$$
Therefore
$$
\|W\|_{\op}^2
=\|I-A^\dagger A\|_{\op}
\leq (\|U\|_{\op}+\|A\|_{\op})\|U-A\|_{\op}
\leq2\epsilon.
$$

Using \(\|B(\cdot)C\|_\Diamond\leq\|B\|_{\op}\|C\|_{\op}\) and
\(\|\Tr_2\|_\Diamond=1\),
$$
\begin{aligned}
\|\Tr_2\circ V-U\|_\Diamond
&\leq
\|A(\cdot)A^\dagger-U(\cdot)U^\dagger\|_\Diamond
+\|\Tr_2(W(\cdot)W^\dagger)\|_\Diamond\\
&\leq
(\|A\|_{\op}+\|U\|_{\op})\|A-U\|_{\op}
+\|W\|_{\op}^2\\
&\leq4\epsilon.
\end{aligned}
$$
\end{proof}

In the formal proof, we need a lemma saying that if one register of a bipartite state is close to a pure state, replacing it with that pure state changes the joint state by only a small amount. We introduce the well-known gentle measurement lemma first.

\begin{lemma}[Gentle measurement lemma~\cite{ON02-gentle-measurement}]\label{lem:gentle-measurement}
    Let $\rho$ be a quantum state, and let $0\preceq M\preceq I$ be a two-outcome measurement. Let $\epsilon=\Tr((1-M)\rho)<1$ be the rejecting probability of this measurement. Then
    $$\left\Vert\rho-\sqrt{M}\rho\sqrt{M}\right\Vert_1\leq 2\sqrt{\epsilon}\,.$$
\end{lemma}

\begin{lemma}\label{lem:sqrt-loss-swapping}
    Suppose $\rho_{AB}$ is a bipartite state. Let $\ket{\phi}$ be a state on register A, such that
    $$\left\Vert \Tr_B(\rho_{AB})-\ketbra{\phi}{\phi}\right\Vert_1\leq \epsilon\,.$$
    Then
    $$\left\Vert \rho_{AB}-\ketbra{\phi}{\phi}\otimes \Tr_A(\rho_{AB})\right\Vert_1\leq O(\sqrt{\epsilon})\,.$$
\end{lemma}

\begin{proof}
    Let $M=\ketbra{\phi}{\phi}\otimes I_B$. Lemma~\ref{lem:gentle-measurement} states
    $$||\rho_{AB}-M\rho_{AB}M||_1\leq O(\sqrt{\epsilon})\,.$$
    Let $\sigma_{B}=\Tr_A(M\rho_{AB}M)$ be the subnormalized state on register B. We have
    $$\Tr_A(\rho_{AB})-\sigma_B\succeq 0\,,$$
    and $\Tr(\Tr_A(\rho_{AB})-\sigma_B)$ is the rejecting probability of measurement $M$. Thus
    $$||\ketbra{\phi}{\phi}\otimes \Tr_A(\rho_{AB})-M\rho_{AB}M||_1\leq \Tr(\Tr_A(\rho_{AB})-\sigma_B)=\epsilon\,.$$
    and the lemma follows by a triangle inequality.
\end{proof}



We also use the following lemma, which works as a workaround to determine the sum of the largest eigenvalues of an Hermitian operator.

\begin{lemma}\label{lem:holder-workaround}
    Suppose $M,\sigma$ are positive semidefinite Hermitian matrices, such that for some $d\geq 1$,
    $$\Tr(\sigma)=1\quad\text{and}\quad 0\preceq \sigma\preceq \frac{I}{d}\,.$$
    Then, for all $s\geq 1$,
    $$\Tr(M\sigma)\leq \left(\frac{\Tr M^s}{d}\right)^{1/s}\,.$$
\end{lemma}
\begin{proof}
    Choose an appropriate basis $\{\ket{j}\}$ and let
    $$M=\sum_{j}\lambda_j\ketbra{j}{j}\,.$$
    Set $p_j=\braket{j|\sigma|j}$. Then
    $$\sum_{j}p_j=1\,\quad\text{and}\quad p_j\leq \frac{1}{d}\,.$$
    Thus, by Jensen's inequality,
    $$(\Tr(M\sigma))^s=\left(\sum_{j}p_j\lambda_j\right)^s\leq \sum_j p_j\lambda_j^s\leq \frac{1}{d}\sum_{j}\lambda_j^s=\frac{1}{d}\Tr(M^s)\,.$$
    Taking $s$-th root on both sides completes the proof.
\end{proof}

\subsection{Time evolution of a Hamiltonian}

The following lemma says that, evolving a Hamiltonian for a short time and then projecting back onto the original subspace approximates the evolution of the effective Hamiltonian on that subspace, with the error quadratic in the evolution time. This quadratic dependence is essential to the quantum Zeno simulation used in our proof.
\begin{lemma}\label{lem:quantum-zeno}
    Let $V:\mathbb C^n\to \mathbb C^m$ be an isometry, and $M:\mathbb C^m\to \mathbb C^m$ be a Hermitian matrix satisfying $0\preceq M\preceq I$. Let $H=V^\dagger MV$. Suppose $\delta\in (0,1)$. Then
    $$\left\Vert e^{i\delta H}-V^\dagger e^{i\delta M}V\right\Vert_{\text{op}}\leq O(\delta^2)\,.$$
\end{lemma}

\begin{proof}
    We have
    $$e^{i\delta H}=I+i\delta H+\sum_{k\geq 2}\frac{(i\delta H)^k}{k!}\,,$$
    $$V^\dagger e^{i\delta M}V=I+i\delta H+\sum_{k\geq 2}\frac{V^\dagger (i\delta M)^kV}{k!}\,.$$
    Since $||M||_{\text{op}}\leq 1$, it yields
    $$||e^{i\delta H}-I-i\delta H||_{\text{op}}\leq \sum_{k\geq 2}\frac{\delta^k}{k!}\,,$$
    $$||V^\dagger e^{i\delta M}V-I-i\delta H||_{\text{op}}\leq \sum_{k\geq 2}\frac{\delta^k}{k!}\,.$$
    Thus
    $$\left\Vert e^{i\delta H}-V^\dagger e^{i\delta M}V\right\Vert_{\text{op}}\leq 2\sum_{k\geq 2}\frac{\delta^k}{k!}\leq O(\delta^2)\,.$$
\end{proof}

We also prove a lemma showing that one can estimate the energy of a Hamiltonian through appropriate time evolution.

\begin{lemma}\label{lem:Hadamard-test}
    Let $\mathcal H$ be a Hilbert space of dimension $d$, $H$ be an Hermitian matrix on $\mathcal H$ satisfying $0\preceq H\preceq I$. Let $\ket{\phi}$ be a state on $\mathcal H$, and denote $p=\braket{\phi|H|\phi}$. Then
    $$p^2\leq \frac{1-\operatorname{Re}\braket{\phi|e^{i\pi H}|\phi}}{2}\leq \frac{\pi^2}{4}p\,.$$
\end{lemma}
\begin{proof}
    Set $H=\sum_{i=1}^d \lambda_i\ketbra{i}{i}$, where $\{\ket{i}\}_{i=1}^d$ is an orthonormal eigenbasis of $\mathcal H$. Let
    $$\ket{\phi}=\sum_{i=1}^d \alpha_i\ket{i}\,.$$
    Then
    $$\frac{1-\operatorname{Re}\braket{\phi|e^{i\pi H}|\phi}}{2}=\frac{1-\sum_{i=1}^d\cos(\pi\lambda_i)|\alpha_i|^2}{2}=\sum_{i=1}^d \sin^2\left(\frac{\pi\lambda_i}{2}\right)|\alpha_i|^2\,.$$
    We have
    $$\lambda_i^2\leq \sin^2\left(\frac{\pi\lambda_i}{2}\right)\leq \frac{\pi^2}{4}\lambda_i^2\leq \frac{\pi^2}{4}\lambda_i\,.$$
    Therefore,
    $$p^2=\left(\sum_{i=1}^d \lambda_i|\alpha_i|^2\right)^2\leq \sum_{i=1}^d \lambda_i^2|\alpha_i|^2\leq \frac{1-\operatorname{Re}\braket{\phi|e^{i\pi H}|\phi}}{2}\,.$$
    $$\frac{1-\operatorname{Re}\braket{\phi|e^{i\pi H}|\phi}}{2}\leq \frac{\pi^2}{4}\sum_{i=1}^d \lambda_i|\alpha_i|^2=\frac{\pi^2}{4}p\,.$$
\end{proof}

\subsection{Properties of the Haar Measure}
We will need the strong concentration of the Haar measure, as stated in \cite{ledoux2001concentration}.

\begin{lemma}[Lévy’s lemma \cite{ledoux2001concentration}]\label{Levy's Lemma}
    Let $d\in\NN$ and $f:\mathbb S^{d}\to \mathbb{R}$ be an $L$-Lipschitz function with respect to the Euclidean distance in the unit sphere $\mathbb S^d$. Then for every $\epsilon>0$, it holds that
    $$\Pr_{x\sim\mathbb S^d}\left[\Big|f(x)-\EE_{z\sim \mathbb S^{d}}f(z)\Big|\geq \epsilon\right]\leq 4\exp\left(-\frac{d\epsilon^2}{9\pi^3L^2}\right)\,.$$
    Here, $x\sim\mathbb S^d$ denotes that $x$ is uniformly chosen on $\mathbb S^d$.
\end{lemma}

Note that the space of complex (resp.\ real) quantum states of norm $1$ in a $d$-dimensional Hilbert space is $\mathbb S^{2d-1}$ (resp.\ $\mathbb S^{d-1}$). In order to use Lemma~\ref{Levy's Lemma}, we prove that the accepting probability of any decision quantum algorithm that takes as input $t$ copies of a quantum state is $\sqrt{t}$-Lipschitz.

\newcommand{\adv}{{\mathcal A}}

\begin{lemma} \label{Pr of accepting is Lipschitz}
    Let $\adv$ be a quantum algorithm that takes as input $t$ copies of a quantum state $\ket{\psi}$ in $\mathbb C^d$. Then, the function $f:\CC^{d}\to \RR$ defined by $f(\ket{\psi})=\Pr[\adv(\ket{\psi}^{\otimes t})=1]$ is $\sqrt{t}$-Lipschitz with respect to the Euclidean distance in $\mathbb C^d$.
\end{lemma}

\begin{proof}
    Let $\ket{\psi},\ket{\phi}$ be two pure states in $\mathbb C^d$. The probability of distinguishing $\ketbra{\psi}{\psi}^{\otimes t}$ from $\ketbra{\phi}{\phi}^{\otimes t}$ is bounded by their trace distance. Hence, we have
    \begin{equation}\label{eq:Pr of accepting is Lipschitz}
        |f(\ket{\psi})-f(\ket{\phi})|\leq \frac{1}{2}\big\|\ketbra{\psi}{\psi}^{\otimes t}-\ketbra{\phi}{\phi}^{\otimes t}\big\|_1\,.
    \end{equation}
    
    Since $\ket{\psi},\ket{\phi}$ are pure states, the trace distance is 
    $$\frac{1}{2}\big\|\ketbra{\psi}{\psi}^{\otimes t}-\ketbra{\phi}{\phi}^{\otimes t}\big\|_1=\sqrt{1-|\braket{\psi|\phi}|^{2t}}\,.$$
    Since, $0\leq |\braket{\psi|\phi}|\leq 1$,
    $$1-|\braket{\psi|\phi}|^{2t}=(1-|\braket{\psi|\phi}|^{2})\sum_{i=0}^{t-1} |\braket{\psi|\phi}|^{2i}\leq t(1-|\braket{\psi|\phi}|^{2})\,,$$
    and we also have that
    $$1-|\braket{\psi|\phi}|^{2}\leq 2-2\mathrm{Re(\braket{\psi|\phi})}=\|\ket{\psi}-\ket{\phi}\|^2\,.$$
    Hence, \eqref{eq:Pr of accepting is Lipschitz} yields
    $$|f(\ket{\psi}-f(\ket{\phi})|\leq \sqrt{t\|\ket{\psi}-\ket{\phi}\|^2}= \sqrt{t}\|\ket{\psi}-\ket{\phi}\|\,,$$
    which implies that $f$ is $\sqrt{t}$-Lipschitz with respect to the Euclidean distance of $\mathbb C^d$.
\end{proof}

\subsection{Property Tests}
Our separation will rely on a variant of the \emph{quantum OR Lemma}. Informally, the following theorem guarantees that there exists a quantum algorithm that given a set of projections and a single copy of some quantum state $\ket{\phi}$, decides whether $\ket{\phi}$ has significant overlap with one of the projections, or if it has small overlap with all the projections.

\begin{theorem}[\cite{WB22}]\label{thm:OR-test}
Given a set of projectors $\{0\preceq P_i\preceq I\}_{i=1}^N$, and a state $\ket{\phi}$, define $p_i=\tr({P_i}\ketbra{\phi}{\phi})$. Define $$p_\downarrow=\max_{1\leq i\leq N} p_i,\ p_\uparrow=\sum_{i=1}^N p_i\,.$$ 
Then there exists an algorithm $\mathcal A$ such that
$$\frac{p^2_\downarrow}{7}\leq \Pr[\mathcal A(\ket{\phi})=1]\leq 2p_\uparrow\,.$$
Furthermore, if each two-outcome measurement $\{P_i, I-P_i\}$ can be implemented in time $T$ with $S$ ancilla qubits, then the algorithm $\mathcal A$ can be implemented in time $O(NT)$ with at most $\poly(N,S)$ ancilla qubits.
\end{theorem}
\begin{proof}
The algorithm runs as follows: it repeats $N$ times, each time it chooses a random measurement $P_i$ and performs the measurement on the current state. If any of the measurements accepts, it outputs $1$. If none of the measurements accepts, it outputs $0$. The proof of the bound on the probability of accepting can be found in~\cite[Theorem 29]{WB22}.
\end{proof}

Next, we describe the generalized SWAP test that is used in our paper, which is the following permutation test.

\begin{theorem}[\cite{KNY08-perm}]\label{thm:permutation-test}
    Let $n,d\in \mathbb N_+$. Let $H$ be the Hilbert space $\mathbb C^d$. For each $\pi\in S_n$, define $U_\pi$ as the following unitary
    $$U_\pi:\ket{\phi_1}\otimes\cdots\otimes \ket{\phi_n}\mapsto \ket{\phi_{\pi(1)}}\otimes \cdots \otimes\ket{\phi_{\pi(n)}}\,,$$
    where $\ket{\phi_1},\cdots,\ket{\phi_n}\in \mathcal H$.
    
    There is an algorithm that implements the projector $\Pi_{\mathrm{sym}}^{(n,d)}$ on $\mathcal H^{\otimes n}$, where
    $$\Pi_{\mathrm{sym}}^{(n,d)}:=\frac{1}{n!}\sum_{\pi\in S_n}U_\pi\,.$$
    Moreover, this algorithm runs in time $\poly(n,\log d)$.
\end{theorem}

\begin{lemma}\label{lem:eigenvalue-Pi_sym}
    Let $\ket{\xi}$ be an arbitrary state in $\mathcal H=\mathbb C^{d}$. Let $p,q\in \mathbb N_+$. Let $\Pi_{\text{sym}}^{(p+q,d)}$ be the permutation test as defined in Theorem~\ref{thm:permutation-test}.
    Define
    \[ Q=Q_{p,q,d}:= \left(I\otimes\bra{\xi}^{\otimes q}\right) \Pi_{\mathrm{sym}}^{(p+q,d)} \left(I\otimes\ket{\xi}^{\otimes q}\right)\,. \]
    Then, there exist a subspaces $\{W_j\}_{0\leq j\leq p}$, such that
    \begin{enumerate}
        \item The dimension of $W_j$ is
        \[\dim W_j=\binom{d+j-2}{j}\,.\]
        \item Each $W_j$ is an eigenspace of $Q$ with eigenvalue
        $$\lambda_j=\frac{p!(p+q-j)!}{(p-j)!(p+q)!}\leq \left(\frac{p}{p+q}\right)^j\,.$$
        \item $Q$ vanishes on the orthogonal complement of $\bigoplus_{j=0}^p W_j$.
    \end{enumerate}

\end{lemma}

\begin{proof}
    We extend $\ket{\xi}$ to an orthonormal basis $\{\ket{0}=\ket{\xi},\ket{1},\ldots,\ket{d-1}\}$ of $\mathcal H$. For every $0\leq j\leq p$, let $W_j\subseteq \mathcal H^{\otimes p}$ be the subspace spanned by all vectors of the form
\[
\Pi_{\text{sym}}^{(p,d)}\left(
\ket{\xi}^{\otimes (p-j)}
\otimes
\ket{x_1}\otimes\cdots\otimes\ket{x_j}
\right)\,,
\]
where $1\leq x_1\leq x_2\leq\cdots\leq x_j\leq d-1$. Note that $\bigoplus_{j=0}^p W_j=\operatorname{Sym}^p \mathcal H$ and we also have that.
%
\[
    \dim W_j
    =
    \binom{d+j-2}{j}\,.
\]

    It remains to prove item (2).
    For every $1\leq x_1\leq x_2\leq\cdots\leq x_j\leq d-1$ and every $\pi\in S_p$, we have 
    \begin{multline*}
        \Pi_{\mathrm{sym}}^{(p+q,d)}U_\pi \left(
\ket{\xi}^{\otimes (p-j)}
\otimes
\ket{x_1}\otimes\cdots\otimes\ket{x_j}
\right)\ket{\xi}^{\otimes q}\\=
\frac{1}{(p+q)!}
\sum_{\sigma\in S_{p+q}}U_\sigma \ket{x_1}\otimes\cdots\otimes\ket{x_j}\otimes \ket{\xi}^{\otimes (p+q-j)}
    \end{multline*}
    which is independent of $\pi$. There are $\frac{p!}{(p-j)!}\cdot (p+q-j)!$ permutations $\sigma$ that map the numbers in $[j]$ to a subset of $[p]$, and thus, each $W_j$ is an eigenspace of $Q$, with eigenvalue
    \[ \lambda_j= \frac{p!(p+q-j)!}{(p-j)!(p+q)!}=\frac{\binom{p+q-j}{q}}{\binom{p+q}{q}}\,. \]
\end{proof}

\subsection{Reflection-to-state simulation}
As suggested by~\cite{Goldin-Zhandry}, we use the following standard theorem to simulate a unitary reflection oracle using copies of the underlying state.
\begin{theorem}[\cite{JLS}]\label{thm:JLS-simulation}
    Let $\ket{\psi}$ be a quantum state. Define the oracle $\mathcal O_{\ket{\psi}}:=I-2\ketbra{\psi}{\psi}$ be the reflection about $\ket{\psi}$. Let $\ket{\xi}$ be a state not necessarily independent of $\ket{\psi}$. Let $\mathcal A^{\mathcal O_{\ket{\psi}}}$ be an oracle algorithm that makes $q$ queries to $\mathcal O_{\ket{\psi}}$. For any integer $l>0$, there is a quantum algorithm $\mathcal B$ that makes no queries to $\mathcal O_{\ket{\psi}}$ such that
    $$\left\Vert\mathcal A^{\mathcal O_{\ket{\psi}}}(\ket{\xi})-\mathcal B(\ket{\psi}^{\otimes l}\otimes \ket{\xi})\right\Vert_1\leq O\left(\frac{q}{\sqrt{l}}\right)\,.$$
    Moreover, the running time of $\mathcal B$ is in polynomial in that of $\mathcal A$ and $l$.
\end{theorem}

\section{Oracle Separation Between Shorter and Longer Length 1PRS}\label{sec: oracle separation}
We will describe a unitary oracle $\mathcal O$ relative to which 1PRS with output length $m(n)=\lfloor \gamma n\rfloor$ exist, but 1PRS with output length $m(n)=\lceil \delta n\rceil$ do not, for any $1<\gamma<\delta$.

Let $\{\ket{\psi_i}\}_{i\in \mathbb N}$ be a family of \emph{real} quantum states, where $\ket{\psi_i}$ is an $i$-qubit state. We will call $\{\ket{\psi_i}\}_{i\in \mathbb N}$ the CHRS states. We now define an oracle $\mathcal O$ with respect to this family of quantum states. For clarity, we think of $\mathcal O$ as a pair of oracles $\mathcal O=(\mathcal O_1,\mathcal O_2)$. Eventually, we will show that for every $1<\gamma<\delta$, with probability $1$ over sampling each state in the family from the \emph{real} Haar measure, $\mathcal O$ separates 1PRS with output length $m(n)=\lfloor \gamma n\rfloor$ from 1PRS with output length $m(n)=\lceil \delta n\rceil$.

Since the possible oracles $\mathcal O$ are in one-to-one correspondence with the family of states $\{\ket{\psi_i}\}_{i\in \mathbb N}$, we endow the set of all possible oracles with the Haar measure of \emph{real} quantum states. Therefore, whenever we say ``the probability over $\mathcal O$", we refer to the Haar measure over the family of the states $\{\ket{\psi_i}\}_{i\in \mathbb N}$.

\paragraph{Separating oracle.} Our separating oracle depends on the following parameters, which we
are free to choose subject to the following constraints:
\begin{itemize}
    \item Real numbers $1<\gamma<\delta$; and
    \item A function $\kappa:\NN\to\NN$, such that for every $i\in \NN$ there are infinitely many $j\in \NN$ such that $\kappa(j)=i$.
\end{itemize}
For brevity, for $i\in [4]$, set $\epsilon_i=\gamma+i(\delta-\gamma)/5$. We consider the sequences $\{R_s\}_{s\in \NN}$, $R_0=0$, and $\{\Lambda_s\}_{s\in \NN^+}$ which are recursively defined as:
\begin{equation}
    \Lambda_s=2^{R_{s-1}}\quad\text{and}\quad
    R_s=\Lambda_s^{7\kappa(s)}\,.
\end{equation}

Also, in the interest of the reflection-to-state simulation (Theorem~\ref{thm:JLS-simulation}), we define
$$\ket{\psi_i^-}=\frac{1}{\sqrt{2}}\ket{0}\ket{0^{i}}-\frac{1}{\sqrt{2}}\ket{1}\ket{\psi_i}\,.$$

We now start the formal definition of the separating oracle.

\begin{definition}\label{def:oracle}
We define the oracle $\mathcal O=(\mathcal O_1,\mathcal O_2)$ as follows.
\begin{itemize}
    \item $\mathcal O_1$ is a family of unitaries $\{\mathcal O_{1,i}\}_{i\in \mathbb N}$, where each $\mathcal O_{1,i}$ is the unique unitary of dimension $2^{i+1}$ that swaps the states $\ket{0^{i+1}}$ and $\ket{1}\ket{\psi_i}$, and acts as identity on the orthogonal complement. Specifically,
    \begin{align*}
        \mathcal O_{1,i}={}&I-\ketbra{0^{i+1}}{0^{i+1}}-\ketbra{1}{1}\otimes \ketbra{\psi_i}{\psi_i}+\ketbra{0^{i+1}}{1}\bra{\psi_i}+\ket{1}\ketbra{\psi_i}{0^{i+1}}\\
        ={}& I-2\ketbra{\psi_i^-}{\psi_i^-}\,.
    \end{align*}
    \item $\mathcal O_2$ is a series of unitaries $\{\mathcal O_{2,\alpha}\}_{\alpha\in \mathbb N}$, which acts as the identity unless the subscript $\alpha$ encodes the following data:
    \begin{itemize}
        \item A number $n$ expressed in unary as $1^n$, intended to denote the input length;
        \item A number $L$ expressed in unary as $1^L$, intended to denote the resource bound, such that $L\in \{R_s\}_{s\in \NN}$;
        \item The description of a Turing machine $M$, which runs in $2^{2^L}$ steps and outputs the description of a unitary quantum circuit $C$ that operates on $n+2^L+\sum_{i=1}^L iL\cdot 2^{i/\epsilon_3}$ qubits.
    \end{itemize}
    For such a subscript $\alpha=(1^n,1^L,M)$, let
    $$\ket{\Psi}=\left(\bigotimes_{i=1}^L \ket{\psi_i^-}^{\otimes L\cdot2^{i/\epsilon_3}}\right)\otimes \ket{0}^{\otimes 2^L}\,$$
    and
    $$H=(I_{2^n}\otimes \bra{\Psi})C^\dagger(\ketbra{1}{1}\otimes I)C(I_{2^n}\otimes \ket{\Psi})\,.$$
    be a Hermitian matrix that acts on $n$ qubits\footnote{Intuitively, $H$ is a two-value POVM measurement corresponding to the following algorithm: on input $\ket{\phi}$, it applies $C$ on $\ket{\phi}\otimes \ket{\Psi}$, and measures the first qubit in the computational basis.}.
    Then, define the unitary
    $$\mathcal O_{2,\alpha}:=e^{i\pi H}\,.$$
\end{itemize}
\end{definition}

\begin{remark}
With regard to~\cite{Zha25-model-unitary-oracles}, throughout the proof, our oracle model allows controlled access and access to the inverse and conjugate. In particular, $\mathcal O_{1,i}$'s are self-inverse and self-conjugate.

    We would like to stress the importance of allowing conjugate access to the oracle. As pointed out by~\cite{Zha25-model-unitary-oracles}, there is a simple construction of stretching the output length of a 1PRS by $1$ bit, assuming that it satisfies an additional assumption. Crucially, the construction in~\cite{Zha25-model-unitary-oracles} requires accessing the conjugate of the 1PRS generator. Allowing conjugate access to the oracle, the separation proven in our paper shows that similar constructions cannot be repeated indefinitely.
\end{remark}

\begin{remark}
    What is important in the definition above is that $\mathcal O_2$ is placed ``sparsely". That is, the valid resource bounds $L$ are placed sparsely, so that for each candidate generator, at certain security parameters $n$, the generator can only access $\mathcal O_2$ with $L=O(\log n)$, but the adversary may access $\mathcal O_2$ with $L=\poly(n)$. As sketched in Section~\ref{sec: tech overview}, this creates the desired imbalance between the generator and the receiver. We explain this in more detail in Section~\ref{sec:longer-1PRS}.
\end{remark}

\subsection{A simulation theorem for the separating oracle}\label{sec:hamiltonian-simulation}

\begin{theorem}\label{thm:hamiltonian-simulation}
    Let $G^{(\cdot)}$ be a unitary oracle circuit that may query the oracle $\mathcal O$\footnote{Including controlled queries and queries to the inverse and conjugate.}, runs in $q$ steps, and queries oracles of length at most $\ell$. Then, for every $\eta\in \mathbb N_+$, there exists a quantum simulation algorithm $S$ that is independent of the CHRS states and does not query $\mathcal O$, and uses $\eta^2+\eta q\ell\cdot 2^{i/\epsilon_3}$ copies of the states $\ket{\psi_i^-}$ for each $i\in[q]$ and $2^\ell q\eta$ ancillas initialized to $\ket{0}$, such that for every state $\ket{\xi}$,
    $$\left\Vert S\left(\ket{\xi}\otimes\left(\bigotimes_{i=1}^\ell \ket{\psi_i^-}^{\otimes(\eta^2+\eta q\ell\cdot 2^{i/\epsilon_3})}\right)\otimes \ket{0}^{\otimes 2^\ell q\eta}\right) - G^\OO\ket{\xi}\right\Vert_1\leq O\left(\frac{q}{\eta}\right)\,.$$
    Moreover, $S$ runs in time $\poly(2^{2^\ell}q ,\eta)$, and does not query the oracle $\mathcal O$.
\end{theorem}

\begin{proof}
    The simulation goes as follows. \begin{itemize}
        \item First, it replaces all oracle calls to $\mathcal O_{1,i}$ with the reflection-to-state simulation algorithm (Theorem~\ref{thm:JLS-simulation}), using $\eta^2$ copies of the state $\ket{\psi_i^-}$.
        \item Then, for each oracle query to $\mathcal O_{2,\alpha}$ with a nontrivial subscript $\alpha=(1^n,1^L,M)$, it replaces $e^{\pm i\pi H}$ (i.e.\ a forward or inverse query) with
        $$\left(\Tr_{\mathsf{aux}}\circ C^\dagger e^{\pm i\pi \ketbra{1}{1}\otimes I/\eta}C(I_{2^n}\otimes \ket{\Psi})\right)^\eta\,,$$
        and the transpose query $(e^{\pm i\pi H})^T=e^{\pm i\pi \overline{H}}$ with
        $$\left(\Tr_{\mathsf{aux}}\circ \overline{C^\dagger} e^{\pm i\pi \ketbra{1}{1}\otimes I/\eta}\overline{C}(I_{2^n}\otimes \ket{\Psi})\right)^\eta\,,$$
        where $\ket{\Psi}$ is as in Definition~\ref{def:oracle}, $\Tr_{\mathsf{aux}}$ denotes the channel defined as tracing out the registers corresponding to $\ket{\Psi}$, and $C$ is the circuit output by $M$. Note that $\ket{\Psi}$ is already a real state. 
        
        In other words, to replace $\mathcal O_{2,\alpha}$, the simulator repeats the following for $\eta$ times: first prepares the auxiliary state $\ket{\Psi}$, then applies the unitary $C^\dagger \left((\ketbra{0}{0}+e^{\pm i\pi/\eta}\ketbra{1}{1})\otimes I\right)C$ or its complex conjugate on the working register and $\ket{\Psi}$, and then traces out the auxiliary register.
    \end{itemize}
    Theorem~\ref{thm:JLS-simulation} proves that the first step introduces a $O(q/\eta)$ loss. To analyze the loss in the second step, we only demonstrate the analysis for the replacement of $e^{i\pi H}$; the computations for the inverse and conjugate queries are essentially equivalent. We use Lemma~\ref{lem:quantum-zeno}:
    $$\left\Vert (I_{2^n}\otimes \bra{\Psi})C^\dagger e^{i\pi \ketbra{1}{1}\otimes I/\eta}C(I_{2^n}\otimes \ket{\Psi})-e^{i\pi H/\eta}\right\Vert_{\text{op}}\leq O\left(\frac{1}{\eta^2}\right)\,.$$
    By a hybrid argument,
    $$\left\Vert \left((I_{2^n}\otimes \bra{\Psi})C^\dagger e^{i\pi \ketbra{1}{1}\otimes I/\eta}C(I_{2^n}\otimes \ket{\Psi})\right)^\eta-e^{i\pi H}\right\Vert_{\text{op}}\leq O\left(\frac{1}{\eta}\right)\,.$$
    Applying Lemma~\ref{lem:trace-out-vs-measurement} yields
    $$\left\Vert \left(\Tr_{\mathsf{aux}}\circ C^\dagger e^{i\pi \ketbra{1}{1}\otimes I/\eta}C(I_{2^n}\otimes \ket{\Psi})\right)^\eta-e^{i\pi H}\right\Vert_\Diamond\leq O\left(\frac{1}{\eta}\right)\,.$$
    Therefore, each replacement introduces a $O(1/\eta)$ loss, so $q$ replacements introduce at most $O(q/\eta)$ loss.

    Notice that for the controlled queries to both $\mathcal O_1,\mathcal O_2$ (and potentially to their inverse and conjugate), the simulation algorithm can be constructed as follows: build the simulation circuit for each query without controlling, and convert each gate to a controlled gate.

    Regarding the number of copies used: $G^{(\cdot)}$ can query $\mathcal O_{1,i}$ for $i\in [\ell]$, which is simulated using $\eta^2$ copies of the state $\ket{\psi_i^-}$. Moreover, the simulation uses at most $\eta \ell\cdot 2^{i/\epsilon_3}$ copies of every state $\ket{\psi_i^-}$ for each query to $\mathcal O_{2,\alpha}$, since, by definition of $\ket{\Psi}$, each simulation of $e^{i\pi H/\eta}$ requires $\ell\cdot 2^{i/\epsilon_3}$ copies of $\ket{\psi_i^-}$. Thus, it uses $\eta^2+\eta q\ell\cdot 2^{i/\epsilon_3}$ copies of each $\ket{\psi_i^-}$, for every $i\in [\ell]$.

    Each query to $\mathcal O_{2,\alpha}$ consumes at most $2^\ell\eta$ qubits of ancillas initialized to $\ket{0}$, so $S$ requires only $2^\ell q\eta$ qubits of ancillas. Since the circuit $C$ runs in at most $2^{2^\ell}$ steps for each query to $\mathcal O_2$, the whole process runs in time $\poly( 2^{2^\ell}q,\eta)$.
\end{proof}

\subsection{The existence of a shorter 1PRS}\label{sec:short-1PRS}

\begin{theorem}\label{thm:short-1PRS}
    Under every choice of $1<\gamma<\delta$ in the separating oracle $\mathcal O$, with probability $1$ over the real Haar measure of the CHRS states, there exists a 1PRS generator with output length $m=m(n)=\lfloor \gamma n\rfloor$.\footnote{Note that this proof works for any choice of $\kappa$ satisfying the constraint in its definition.}
\end{theorem}

\begin{proof}
    The proposed 1PRS generator is $G^\OO(k)=(Z^k\otimes I)\ket{\psi_m}$, where $\ket{\psi_m}$ can be obtained by tracing out the first qubit of $\mathcal O_{1,m}\ket{0^{m+1}}$. 

    

    Consider an efficient nonuniform adversary $\adv^\OO(\cdot;c)$ which makes at most $q$ queries to the adversary and with advice length $|c|\leq q$, where $q$ is a polynomial in $n$.
    By Theorem~\ref{thm:hamiltonian-simulation}, choosing $\eta=q\cdot 2^{(\epsilon_1-\gamma)n/2\epsilon_3}$, the adversary can be simulated with precision $O(1/2^{(\epsilon_1-\gamma)n/2\epsilon_3})$, with
    $$t:=q^2\cdot 2^{(\epsilon_1-\gamma)n/\epsilon_3}+
    q^{2}2^{(\epsilon_1-\gamma)n/2\epsilon_3}
    \cdot 2^{m/\epsilon_3}
    \leq2^{((\epsilon_2-\gamma)n+m)/\epsilon_3}\leq 2^{\epsilon_2n/\epsilon_3}$$
    copies of $\ket{\psi_m^-}$ and some states that are independent of $\ket{\psi_m}$. Let $\mathcal S$ be the simulator of $\mathcal A^\OO$ and let 
$$\mathrm{adv}(\mathcal S,c)=\mathop{\Pr}_{k\gets\{0,1\}^n}[\mathcal S(\ket{\phi_k}\otimes \ket{\psi_m^-}^{\otimes t};c)=1]-\mathop{\Pr}_{\ket{\phi}\sim \mu_{2^m}^{\mathbb R}} [\mathcal S(\ket{\phi}\otimes \ket{\psi_m^-}^{\otimes t};c)=1]\,.$$
    Theorem~\ref{thm:real-haar-stat-close} states that for every $t\in \mathbb N$,
    $$\left\Vert \EE_{\substack{k\gets \{0,1\}^n\\\ket{\psi_m}\sim \mu_{2^m}^{\mathbb R}}}(Z^k\otimes I)\ketbra{\psi_m}{\psi_m}(Z^k\otimes I)\otimes \ketbra{\psi_m^-}{\psi_m^-}^{\otimes t}-\EE_{\ket{\psi_m}\sim \mu_{2^m}^{\mathbb R}}\frac{I}{2^m}\otimes \ketbra{\psi_m^-}{\psi_m^-}^{\otimes t}\right\Vert_1\leq O\left(\frac{t}{2^n}\right)\,.$$
    Using $t= 2^{\epsilon_2n/\epsilon_3}$ we obtain that
    \begin{equation}\label{eq:exp adv small}
        \left|\EE_{\ket{\psi_m}\sim \mu_{2^m}^{\mathbb R}}\mathrm{adv}(\mathcal S,c)\right|\leq O\left(\frac{1}{2^{(1-\epsilon_2/\epsilon_3)n}}\right)\,.
    \end{equation}

To show that $G^\OO$ is a 1PRS generator, it is required to fix a separating oracle $\OO$ such that $\mathcal S$ cannot achieve a non-negligible advantage. We use the following idea inspired by the proof of \cite[Theorem 30]{Kre21}.

By Lemma~\ref{Pr of accepting is Lipschitz} and the fact that the expectation of a Lipschitz function is also Lipschitz, the function $\mathop{\Pr}_{k\gets\{0,1\}^n}[\mathcal S(\ket{\phi_k}\otimes \ket{\psi_m^-}^{\otimes t};c)=1]$ is a $\sqrt{{t+1}}$-Lipschitz function in $\ket{\psi_m^-}$. Since $\ket{\psi_m^-}$ is simply a linear shift of $\ket{\psi_m}$, it is a $\sqrt{(t+1)/2}$-Lipschitz function in $\ket{\psi_m}$. The same applies to $ \mathop{\Pr}_{\ket{\phi}\sim \mu_{2^m}^{\mathbb R}} [\mathcal S(\ket{\phi}\otimes \ket{\psi_m^-}^{\otimes p};c)=1]$, so $\mathrm{adv}(\mathcal S,c)$ is a $\sqrt{2(t+1)}$-Lipschitz function in $\ket{\psi_m}$, where $t= 2^{\epsilon_2n/\epsilon_3}$. Recall that $m>n$, so Levy's Lemma~\ref{Levy's Lemma} implies
\begin{align}
    \Pr_{\ket{\psi_m}\sim \mu_{2^m}^{\mathbb R}}\Big[\Big|\mathrm{adv}(\mathcal S,c)&-\EE_{\ket{\psi_m}\sim \mu_{2^m}^{\mathbb R}}\mathrm{adv}(\mathcal S,c)\Big|\geq 2^{-(1-\epsilon_2/\epsilon_3)n/4}\Big]\notag
    \\
    &\leq 4\exp\left(-\frac{2\cdot 2^m\cdot 2^{-(1-\epsilon_2/\epsilon_3)n/2}}{9\pi^3\cdot 2(t+1)}\right)\leq O\left(\exp\left(-2^{(1-\epsilon_2/\epsilon_3)n/3}\right)\right)\,.\label{eq:adv - E adv is small}
\end{align}
Hence, Equations~\eqref{eq:exp adv small} and~\eqref{eq:adv - E adv is small} imply that for every $c\in \{0,1\}^q$,
\begin{equation}\label{eq:adv is small whp}
    \Pr_{\ket{\psi_m}\sim \mu_{2^m}^{\mathbb R}}\Big[\Big|\mathrm{adv}(\mathcal S,c)\Big|\geq 2^{-(1-\epsilon_2/\epsilon_3)n/4}
+\Omega(2^{-(1-\epsilon_2/\epsilon_3)n})\Big]\leq O\left(\exp\left(-2^{(1-\epsilon_2/\epsilon_3)n/3}\right)\right)\,.
\end{equation}
Let
$$\mathrm{adv}(\mathcal S)=\max_{c\in\{0,1\}^q} \{\mathrm{adv}(S,c)\}\,,$$
then using a union bound, Equation~\eqref{eq:adv is small whp} yields that
\begin{align}
    \Pr_{\ket{\psi_m}\sim \mu_{2^m}^{\mathbb R}}\Big[\Big|\mathrm{adv}(\mathcal S)\Big|\geq 2^{-(1-\epsilon_2/\epsilon_3)n/4}
+\Omega(2^{-(1-\epsilon_2/\epsilon_3)n})\Big]\leq {}& 2^q\cdot O\left(\exp\left(-2^{(1-\epsilon_2/\epsilon_3)n/3}\right)\right)\notag\\
\leq {}& O\left(\exp\left(-2^{(1-\epsilon_2/\epsilon_3)n/4}\right)\right)\,.\label{eq:adv S is small}
\end{align}

Recall that the simulation error is $O(2^{-(\epsilon_1-\gamma)n/2\epsilon_3})$. So, let 
$$\mathrm{adv}(\mathcal A^{\mathcal O})=\mathop{\Pr}_{k\gets\{0,1\}^n}[\mathcal A^{\mathcal O}(\ket{\phi_k})=1]-\mathop{\Pr}_{\ket{\phi}\sim \mu_{2^m}^{\mathbb R}} [\mathcal A^{\mathcal O}(\ket{\phi})=1]\,,$$
then for any $\ket{\psi_m}$,
\begin{equation}\label{eq: advA and advS are close}
    \left|\mathrm{adv}(\mathcal A^{\mathcal O})-\mathrm{adv}(\mathcal S)\right|\leq O(2^{-(\epsilon_1-\gamma)n/2\epsilon_3})\,.
\end{equation}
Thus, combining Equations~\eqref{eq:adv S is small} and~\eqref{eq: advA and advS are close} we obtain that
\begin{multline*}
    \Pr_{\ket{\psi_m}\sim \mu_{2^m}^{\mathbb R}}\Big[\Big|\mathrm{adv}(\adv^\OO)\Big|\geq 2^{-(1-\epsilon_2/\epsilon_3)n/4}
+\Omega(2^{-(1-\epsilon_2/\epsilon_3)n})+\Omega(2^{-(\epsilon_1-\gamma)n/2\epsilon_3})\Big]\\\leq O\left(\exp\left(-2^{(1-\epsilon_2/\epsilon_3)n/4}\right)\right)\,.
\end{multline*}

Since $\sum_{n=1}^{\infty} \exp\left(-2^{(1-\epsilon_2/\epsilon_3)n/4}\right)<\infty$, by Borel--Cantelli Lemma, $\adv^\OO$ achieves an advantage of at most $O(2^{-(\delta-\gamma)n/(20\delta)})$ for all but finitely many $n\in \NN$, with probability $1$ over the choice of the CHRS states. Thus, with probability $1$, every QPT adversary achieves a negligible advantage towards the 1PRS candidate $G$ under our separating oracle $\mathcal O$.
\end{proof}

\subsection{Non-existence of longer 1PRS}\label{sec:longer-1PRS}

We would like to show that no 1PRS with output length $m(n)=\lceil\delta n\rceil$ exists under our oracle $\OO$. To show this, let $G^\OO$ be a candidate 1PRS generator, such that for any input $k\in\{0,1\}^n$, it runs in $O(n^\tau)$ steps, and queries oracles of length at most $O(n^\tau)$, where $\tau$ is a constant that depends on $G^{(\cdot)}$. To simplify, for every $k\in \{0,1\}^n$, let $G_k$ be the circuit that generates the 1PRS state $\ket{\phi_k}$ corresponding to input $k$, such that
$$G_k\ket{0^{m+a}}=\ket{\phi_k}\ket{\tilde \phi_k}\,,$$
where $a$ is the number of ancillas and $\ket{\tilde \phi_k}$ is some $a$-qubit state.

Our attack for the longer 1PRS will be an OR test on whether the challenge state $\ket{\phi}$ equals $\ket{\phi_k}$ for some $k\in \{0,1\}^n$ or is Haar random. However, our OR test requires exponential time, so we have to implement it using the oracle $\mathcal O_2$. 

Let $S_\tau=\{s\in \NN:\kappa(s)=\lceil\tau\rceil\}$, we will describe an adversary that breaks $G^\OO$ for all security parameters in the set $N_G=\{\Lambda_s\}_{s\in S_\tau}$. Now, let $n\in N_G$ and let $s_n\in \NN$ be such that $\Lambda_{s_n}=n$. Since, $n^\tau <R_{s_n}$ we have that $G^{(\cdot)}$ can make queries to $\OO_2$ of length at most $R_{s_n-1}=\log \Lambda_{s_n}=\log n$. Hence, using Theorem~\ref{thm:hamiltonian-simulation}, we have that for every $\eta\in \NN$, $G^\OO$ can be simulated using $\eta^2+\eta q(\log n)2^{(\log n)/\epsilon_3}\leq \eta^2+\eta q^3$ copies of each state $\ket{\psi_i^-}$, $i\in [q]$ and $qn\eta$ ancillas.




\paragraph{The OR test.} Given the challenge state $\ket{\phi}$, our algorithm runs an OR test with the projectors $\{\Pi_k\}_{k\in \{0,1\}^n}$, where each $\Pi_k$ is defined as follows.

\begin{definition}\label{def:Pi_testers}
Let $G^\OO$ be a candidate 1PRS generator that runs in time $q=\Theta(n^\tau)$, $\tau\geq 1$.\footnote{In general, for an algorithm running in time $O(n^\tau)$, we may pad the running time to $\Theta(n^\tau)$.} Set $\eta=q^2$, $w_i=\eta^2+\eta q^3$, $t_i=(2^{i/\epsilon_3}q-1)w_i$, $i\in [q]$, and consider the following registers:
\begin{itemize}
    \item $\mathsf{A},\mathsf{B}$ are both $m$-qubit registers;
    \item $\mathsf{C}$ is an $a$-qubit register (recall that $a$ is the number of ancillas used in $G$);
    \item $\mathsf{D},\mathsf{E}$ consists of $q$ subregisters. The $i$-th subregister $\mathsf{D}_i$ (resp.\ $\mathsf{E_i}$) is of $iw_i$ qubits (resp.\ $it_i$ qubits);
    \item $\mathsf{F}$ is a $qn\eta$-qubit register; and
    \item $\mathsf{G}$ is a register of $(2^{q}-qn\eta-m-a)$ qubits.
\end{itemize}
For each $k\in \{0,1\}^n$, 
define $\Pi_k$ as a projector on the registers $\mathsf{ABCDEFG}$, corresponding to the following measurement. The measurement uses Theorem~\ref{thm:hamiltonian-simulation} to prepare a simulator $S_k$ for the 1PRS generator $G_k$, with $\eta=q^2$. Denote the unitary part of $S_k$ as $U_k$.
\begin{enumerate}[label=\alph*)]
    \item First, it applies $U_k$ to the registers $\mathsf{BCDF}$.
    \item Then, it applies $U_k^\dagger$ to the registers $\mathsf{ACDF}$.
    \item It measures the registers $\mathsf{A}$, $\mathsf{C}$ and $\mathsf{F}$ in the computational basis.
    \item For each $i\in [q]$, it performs the generalized SWAP test (Theorem~\ref{thm:permutation-test}) on the registers $\mathsf{D}_i$ and $\mathsf{E}_i$, i.e.\ performs the projective measurement $\Pi_{\text{sym}}^{(w_i+t_i,2^i)}$.
    \item $\Pi_k$ accepts if and only if all the measurements in (c) yield $0$ and all tests in (d) accept.
\end{enumerate}
\end{definition}

Crucially, for every $i\in[q]$,
    $$w_i+t_i\leq (q^4+q^{5})q2^{i/\epsilon_3}\leq q^{7} 2^{i/\epsilon_3}\,,$$
    for sufficiently large $n$. Thus, each test $\Pi_k$ can be carried out by an adversary that makes queries of length $O(q^7)=O(n^{7\tau})$.

Define the state
$$\ket{\Phi}:=\ket{\phi}_{\mathsf{A}}\otimes \left(\bigotimes_{i=1}^q \ket{\psi_i^-}^{\otimes (w_i+t_i)}_{\mathsf{D}_i\mathsf{E}_i}\right)\otimes \ket{0}^{\otimes 2^{q}}_{\mathsf{BCFG}}\,.$$
where $\ket{\phi}$ is the challenge state.
The adversary identifies whether $\ket{\phi}$ is pseudorandom or Haar random, by performing a test essentially equivalent to the quantum OR test for all $\{\Pi_k\}_{k\in \{0,1\}^n}$ on the state $\ket{\Phi}$.


By the OR lemma (Theorem~\ref{thm:OR-test}), in order to compute the probability that the OR test accepts, it suffices to compute the probability for each individual test to accept. We split the analysis into the cases that $\ket{\phi}=\ket{\phi_k}$ for some $k\in \{0,1\}^n$ or $\ket{\phi}$ is Haar random. We will prove that, in the former case, there exists an individual test accepting with probability $\Omega(1)$ (Proposition~\ref{prop:generated}); in the latter case, every individual test accepts with probability $o(2^{-n})$ (Proposition~\ref{prop:haar-random}). Then, we combine both results to build an adversary that breaks the longer 1PRS (Theorem~\ref{thm:longer-1PRS}).

\subsubsection{The analysis when the input state is pseudorandom}

\begin{proposition}\label{prop:generated}
    Let $G^\OO$ be a candidate 1PRS generator running in time $q=q(n)=\Theta(n^\tau)$, $\tau\geq 1$ for each security parameter $n$. For every $k\in\{0,1\}^n$, let $$\ket{\Phi_k}=\ket{\phi_k}_{\mathsf{A}}\otimes \left(\bigotimes_{i=1}^q \ket{\psi_i^-}^{\otimes (w_i+t_i)}_{\mathsf{D}_i\mathsf{E}_i}\right)\otimes \ket{0}^{\otimes 2^{q}}_{\mathsf{BCFG}}\,,$$
    where $\ket{\phi_k}$ is the candidate 1PRS state generated by $G_k$. Let $\Pi_k$ be the projector defined in Definition~\ref{def:Pi_testers}. Then
    $$\Tr(\Pi_k\ketbra{\Phi_k}{\Phi_k})\geq 1-O\left(\frac{1}{\sqrt{n^\tau}}\right)\,.$$
\end{proposition}
\begin{proof}
    Let $\ket{\theta_k}$ be the state on the registers $\mathsf{BCDF}$ after applying $U_k$ in step (a), and let $\rho$ be the state on the registers $\mathsf{ACDF}$ right after step (a). By Theorem~\ref{thm:hamiltonian-simulation}, the marginal state on the register $\mathsf{B}$ of $\ket{\theta_k}$ is $O(1/n^\tau)$-close to the state $\ket{\phi_k}$. Lemma~\ref{lem:sqrt-loss-swapping} implies that
    $$\left\Vert \rho_{\mathsf{ACDF}}-\ketbra{\theta_k}{\theta_k}_{\mathsf{BCDF}}\right\Vert _1\leq O\left(\frac{1}{\sqrt{n^\tau}}\right)\,.$$
    Since applying $U_k^\dagger$ on $\ket{\theta_k}$ would return the system back to the state $\ket{\Phi_k}$, for the state of the registers $\mathsf{ACDF}$ after applying $U_k^\dagger$ to $\rho$ in step (b), it holds that
    $$\left\Vert (U_k^\dagger\rho U_k)_{\mathsf{ACDF}}-\ketbra{\Phi_k}{\Phi_k}_{\mathsf{BCDF}}\right\Vert _1\leq O\left(\frac{1}{\sqrt{n^\tau}}\right)\,.$$
    The registers $\mathsf{BCF}$ of the state $\ket{\Phi_k}$ clearly pass all the tests in steps (c), and the registers $\mathsf{D}$ of the state $\ket{\Phi_k}$ passes (d), and thus, the proposition immediately follows from the properties of the trace distance.
\end{proof}

\subsubsection{The analysis when the input state is Haar random}

\begin{proposition}\label{prop:haar-random}
    Let $G^\OO$ be a candidate 1PRS generator running in time $q=q(n)=O(n^\tau)$, $\tau\geq 1$ for each security parameter $n$. For every $m$-qubit state $\ket{\phi}$,  let $$\ket{\Phi}=\ket{\phi}_{\mathsf{A}}\otimes \left(\bigotimes_{i=1}^q \ket{\psi_i^-}^{\otimes (w_i+t_i)}_{\mathsf{D}_i\mathsf{E}_i}\right)\otimes \ket{0}^{\otimes 2^{q}}_{\mathsf{BCFG}}\,.$$
    For every $k\in\{0,1\}^n$, let $\Pi_k$ be the projector defined in Definition~\ref{def:Pi_testers}. Then
    $$\EE_{\ket{\phi}\sim \mu_{2^m}}\Tr(\Pi_k\ketbra{\Phi}{\Phi})\leq (1+O(1/q))\cdot 2^{-m/\epsilon_4}\,,$$
    where $\mu_{2^m}$ denotes the (complex) Haar measure on $m$-qubit pure states.
\end{proposition}

\begin{proof}
    Following the notation in Definition~\ref{def:Pi_testers}, let $\rho$ be the state on the registers $\mathsf{ACDF}$ right after step (a). 
    Recall that the test in Definition~\ref{def:Pi_testers} accepts with probability $1$ if $U_k^\dagger\rho U_k$ is exactly equal to the density matrix of the pure state $$\ket{0^{m+a}}_{\mathsf{AC}}\otimes\left(\bigotimes_{i=1}^q \ket{\psi_i^-}^{\otimes w_i}_{\mathsf{D}_i}\right)\otimes \ket{0}^{\otimes qn\eta}_{\mathsf{F}}\,.$$
    So, let
    \begin{equation}\label{eq:def projector for steps c and d}
        \widetilde\Pi_k=\left(I_{\mathsf{ACDF}}\otimes\left(\bigotimes_{i=1}^q \bra{\psi_i^-}^{\otimes t_i}_{\mathsf{E}_i}\right)\right) \ketbra{0}{0}_\mathsf{AC}^{\otimes (m+a)}\otimes \left(\bigotimes_{i=1}^q\Pi_{\text{sym}\,\,\mathsf{D_iE_i}}^{(w_i+t_i,2^i)}\right)\otimes\ketbra{0}{0}_\mathsf{F}^{\otimes qn\eta}\left(I_{\mathsf{ACDF}}\otimes\left(\bigotimes_{i=1}^q \ket{\psi_i^-}^{\otimes t_i}_{\mathsf{E}_i}\right)\right)
    \end{equation}
    be the operator that corresponds to the tests in steps (c) and (d). Then we can write the acceptance probability as $\Tr(\Pi_k\ketbra{\Phi}{\Phi})=\Tr(\widetilde\Pi_k U_k^\dagger\rho U_k)$.

    By our definition of $\rho$, we have that
    $$\EE_{\ket{\phi}\sim \mu_{2^m}}\left[\rho\right]=\frac{I_\mathsf{A}}{2^m}\otimes \rho'\,,$$
    where $\rho'$ is some state on the registers $\mathsf{CDF}$ that is unentangled with $\mathsf{A}$. In particular, $\rho'\preceq I$, and thus
    $$\EE_{\ket{\phi}\sim \mu_{2^m}}[U_k^\dagger\rho U_k]\preceq \frac{I}{2^m}\,.$$

    Using Lemma~\ref{lem:holder-workaround}, we have that
    \begin{equation}
        \EE_{\ket{\phi}\sim \mu_{2^m}}\Tr(\widetilde\Pi_k U_k^\dagger\rho U_k)\leq \left(\frac{\Tr(\widetilde \Pi_k^{\epsilon_4})}{2^m}\right)^{1/\epsilon_4}\,.\label{eq:bound of acceptance}
    \end{equation}

    Now we compute $\Tr(\widetilde \Pi_k^{\epsilon_4})$. Recall the definition of the operator $Q$ in Lemma~\ref{lem:eigenvalue-Pi_sym} that helps simplify the expression of Equation~\eqref{eq:def projector for steps c and d}:
    $$\widetilde \Pi_k=\ketbra{0}{0}_{\mathsf{ACF}}\otimes \left(\bigotimes_{i=1}^q Q_{w_i,t_i,2^i\, \mathsf{D_i}}\right)\,.$$
    Thus,
    $$\widetilde \Pi_k^{\epsilon_4}=\ketbra{0}{0}_{\mathsf{ACF}}\otimes \left(\bigotimes_{i=1}^q Q_{w_i,t_i,2^i\, \mathsf{D_i}}^{\epsilon_4}\right)\,.$$
    Using Lemma~\ref{lem:eigenvalue-Pi_sym}, we have that
    \begin{align}
        \Tr (Q_{w_i,t_i,2^i}^{\epsilon_4})\leq{}&\sum_{j=0}^{w_i}\binom{2^i+j-2}{j}\left(\frac{w_i}{w_i+t_i}\right)^{\epsilon_4j}\notag\\
        \leq{}&\sum_{j=0}^{\infty}\binom{2^i+j-2}{j}\left(\frac{w_i}{w_i+t_i}\right)^{\epsilon_4j}\notag\\
        \leq{} & \left(1-\left(\frac{w_i}{w_i+t_i}\right)^{\epsilon_4}\right)^{-(2^i-1)}\,,\label{eq:bound on TrQ^epsilon4}
    \end{align}
    where~\eqref{eq:bound on TrQ^epsilon4} follows from the fact that
    $$\sum_{j=0}^\infty \binom{2^i+j-2}{j}x^j=(1+x+x^2+\cdots)^{2^i-1}=(1-x)^{-(2^i-1)}\,.$$
    Thus,
    \begin{align}
        \log \Tr(\widetilde \Pi_k^{\epsilon_4})={} & \sum_{i=1}^q \log \Tr (Q_{w_i,t_i,2^i}^{\epsilon_4})\notag\\
        \leq{} & \sum_{i=1}^q (2^i-1)\log\frac{1}{1-(2^{i/\epsilon_3}q)^{-\epsilon_4}}\notag\\
        \leq{} & 2q^{-\epsilon_4}\sum_{i=1}^q 2^{i(1-\epsilon_4/\epsilon_3)}\,,\label{eq:log Tr Pikepsilon4}
    \end{align}
    where~\eqref{eq:log Tr Pikepsilon4} is because for sufficiently small $x$,
    $$\log\left(\frac{1}{1-x}\right)=x+O(x^2)\leq 2x\,.$$
    Since $\epsilon_4>\epsilon_3$, Equation~\eqref{eq:log Tr Pikepsilon4} yields
    $$\log \Tr (\widetilde \Pi_k^{\epsilon_4})=O(q^{-\epsilon_4})\,,$$
    which implies that
    $$\Tr(\widetilde\Pi_k^{\epsilon_4})=1+O(q^{-\epsilon_4})\,.$$
    Thus, substituting to Equation~\eqref{eq:bound of acceptance}, we obtain that
    $$\EE_{\ket{\phi}\sim \mu_{2^m}}\Tr(\widetilde\Pi_k U_k^\dagger\rho U_k)=(1+O(q^{-\epsilon_4}))^{1/\epsilon_4}\cdot 2^{-m/\epsilon_4}=(1+O(1/q))\cdot 2^{-m/\epsilon_4}\,.$$
\end{proof}

\subsubsection{The OR test}

\begin{theorem}\label{thm:longer-1PRS}
    There does not exist any 1PRS with $m=\lceil \delta n\rceil$ under the oracle $\mathcal O$.
\end{theorem}

\begin{proof}
    Given a 1PRS generator $G^\OO$ that runs in time $q=O(n^\tau)$, let $S_\tau=\{s\in \NN:\kappa(s)=\lceil\tau\rceil\}$. We construct the adversary $\adv^\OO$ as follows. On security parameter $n=\Lambda_s$ for some $s\in S_\tau$, it prepares the Turing machine $M$ that generates the description of the OR test (Theorem~\ref{thm:OR-test}) over all $\{\Pi_k\}_{k\in \{0,1\}^n}$ as defined in Definition~\ref{def:Pi_testers}. Let $\alpha=(m,1^{q^7},M)$.
    The adversary runs the following circuit
    \begin{center}
    \begin{quantikz}
    \lstick{$\ket{0}$} & \gate{\mathsf{H}} & \ctrl{1} & \gate{\mathsf{H}} & \meter{} \\
    \lstick{$\ket{\phi}$} &  & \gate{\mathcal O_{2,\alpha}} &  & 
    \end{quantikz}
    \end{center}
    and outputs the outcome of the measurement.

    Recall that $\mathcal O_{2,\alpha}=e^{i\pi H}$, where $H$ is the POVM measurement described by $M$; that is, the OR test over all $\{\Pi_k\}_{k\in \{0,1\}^n}$. By Proposition~\ref{prop:generated}, Proposition~\ref{prop:haar-random}, and Theorem~\ref{thm:OR-test}, we have that for every $k\in\{0,1\}^n$,
    $$\braket{\phi_k|H|\phi_k}\geq \frac{1}{7}\left(1-O\left(\frac{1}{\sqrt{n^\tau}}\right)\right)^2\geq \frac{1}{7}\left(1-O\left(\frac{1}{\sqrt{n^\tau}}\right)\right)\,,$$
    while
    $$\EE_{\ket{\phi}\sim \mu_{2^m}}\braket{\phi|H|\phi}\leq 2^n\cdot \left((1+O(1/q))\cdot 2^{-m/\epsilon_4}\right)\leq \frac{O(1)}{2^{(\delta/\epsilon_4-1)n}}\leq \negl(n)\,.$$
    One can verify that for every input state $\ket{\phi}$,
    $$\Pr\left[\mathcal A^\OO(\ket{\phi})=1\right]=\frac{1-\operatorname{Re}\braket{\phi|\mathcal O_{2,\alpha}|\phi}}{2}=\frac{1-\operatorname{Re}\braket{\phi|e^{i\pi H}|\phi}}{2}\,.$$
    Hence, Lemma~\ref{lem:Hadamard-test} implies that 
    $$\Pr_{k\in \{0,1\}^n}\left[\mathcal A^\OO(\ket{\phi_k})=1\right]\geq \EE _{k\in \{0,1\}^n}\left[(\braket{\phi_k|H|\phi_k})^2\right]\geq \frac{1}{50}\,,$$
    and
    $$\Pr_{\ket{\phi}\sim \mu_{2^m}}\left[\mathcal A^\OO(\ket{\phi})=1\right]\leq \frac{\pi^2}{4}\left(\EE_{\ket{\phi}\sim \mu_{2^m}}\braket{\phi|H|\phi}\right)\leq \negl(n)\,.$$
    Thus, we obtain a non-negligible distinguishing advantage.
\end{proof}

\section{A Tight Bound on the Security of the \cite{AGL24} 1PRS Construction}\label{sec:short-1PRS-security}
In this section, we prove the following theorem, which improves the security of the \cite{AGL24} 1PRS construction, and gives a tight bound:
\begin{theorem}\label{thm:real-haar-stat-close}
    Let $n,m\in\mathbb N_+$ with $m\geq n+1$. Let $\ket{\psi}$ be a real $m$-qubit state and let
    $$\ket{\psi^-}=\frac{1}{\sqrt{2}}(\ket{0^{m+1}}-\ket{1}\ket{\psi})\,.$$
    Then for every $t\in \NN_+$,
    $$\left\Vert \EE_{\substack{k\gets \{0,1\}^n\\\ket{\psi}\sim \mu_{2^m}^{\mathbb R}}}(Z^k\otimes I)\ketbra{\psi}{\psi}(Z^k\otimes I)\otimes \ketbra{\psi^-}{\psi^-}^{\otimes t}-\EE_{\ket{\psi}\sim \mu_{2^m}^{\mathbb R}}\frac{I}{2^m}\otimes \ketbra{\psi^-}{\psi^-}^{\otimes t}\right\Vert_1\leq O\left(\frac{t}{2^n}\right)\,,$$
    where $\mu_{2^m}^{\mathbb R}$ is the real Haar measure on the $2^m$-dimensional unit sphere.
\end{theorem}
The use of real Haar-random states, as well as the replacement of $\ket{\psi}$ with $\ket{\psi^-}$, is motivated by our need to upgrade the CHRS access model to a unitary-oracle model with conjugate access, as required by~\cite{Zha25-model-unitary-oracles}. Note that the same bound also holds when $\ket{\psi}$ is just a Haar random state.

By the well-known $Z$-twirl, applying a random $Z$ gate on the first $n$ qubits is equivalent to measuring the first $n$ qubits in the computational basis. More formally:
\begin{lemma}[$Z$-twirl]\label{Z^k is equiv to measurement}
    Let $m,n\in \NN_+$ with $m>n$. Then for every $m$-qubit state $\ket{\psi}$ we have that
    $$\EE_{k\gets\{0,1\}^n}(Z^k\otimes I)\ketbra{\psi}{\psi}(Z^k\otimes I)=\sum_{k\in\{0,1\}^n}(\ketbra{k}{k}\otimes I)\ketbra{\psi}{\psi}(\ketbra{k}{k}\otimes I)\,.$$
\end{lemma}
For the sake of exposition, we defer the proof of Lemma~\ref{Z^k is equiv to measurement} to the Appendix~\ref{sec:Appendix}.

For every $k\in \{0,1\}^n$, let $\Pi_k= \ketbra{k}{k}\otimes I$. Using Lemma~\ref{Z^k is equiv to measurement}, in order to derive Theorem~\ref{thm:real-haar-stat-close}, it suffices to prove that 
    \begin{equation}\label{eq:equivalent of theorem statement}
        \left\Vert \EE_{\ket{\psi}\sim \mu_{2^m}^{\mathbb R}}\left[\left(\sum_{k\in\{0,1\}^n}\Pi_k\ketbra{\psi}{\psi}\Pi_k-\frac{I}{2^m}\right)\otimes \ketbra{\psi^-}{\psi^-}^{\otimes t}\right]\right\Vert_1\leq O\left(\frac{t}{2^n}\right)\,.
    \end{equation}
    Let 
    \begin{equation}\notag
        \rho = \EE_{\ket{\psi}\sim \mu_{2^m}^{\mathbb R}}\left[\left(\sum_{k\in\{0,1\}^n}\Pi_k\ketbra{\psi}{\psi}\Pi_k\right)\otimes \ketbra{\psi^-}{\psi^-}^{\otimes t}\right]
        \end{equation}
    and
    \begin{equation}\notag
        \sigma = \EE_{\ket{\psi}\sim \mu_{2^m}^{\mathbb R}}\left[\frac{I}{2^m}\otimes \ketbra{\psi^-}{\psi^-}^{\otimes t}\right]\,.
    \end{equation}
    Then Equation~\eqref{eq:equivalent of theorem statement} is equivalent to $\|\rho-\sigma\|_1\leq O(t/2^n)$. 
    
    As discussed in Section~\ref{sec: tech overview}, in order to bound $\|\rho-\sigma\|_1$, we would like to remove the correlations between the first register of $\rho$ and the rest, without significantly disturbing either of the two states. Hence, we consider the projector
    \begin{equation}\label{eq:def of Pi}
        \Pi=\sum_{k\in\{0,1\}^n} \Pi_k\otimes (I-\Pi_k^-)^{\otimes t}\,,
    \end{equation}
    where $\Pi_k^-=\ketbra{1}{1}\otimes \Pi_k$, for every $k\in\{0,1\}^n$. Then, we can write that
    \begin{equation}\label{eq:bound by 3 terms}
        \|\rho-\sigma\|_1\leq \|\rho-\Pi\rho\Pi\|_1+\|\Pi(\rho-\sigma)\Pi\|_1+\|\sigma-\Pi\sigma\Pi\|_1\,.
    \end{equation}
    We will show that each term of Equation~\eqref{eq:bound by 3 terms} is bounded by $O(t/2^n)$. 
    
    In our proof, we will need the following lemmas:
\begin{lemma}[\cite{repec:eee:csdana:v:22:y:1996:i:5:p:568-569}]\label{lem:chi^2 moment}
    Let $n,r\in \NN_+$. Let $X_1,X_2,\ldots,X_n\sim \mathcal{N}(0,1)$, and $Y=\sum_{i=1}^n X_i^2$, then $\EE[Y^r]=n(n+2)\cdots(n+2r-2)$.
\end{lemma}

\begin{lemma}\label{lem:bound of Ex_i^2}
    Let $n,m\in \NN_+$ with $m>n$. Let $\ket{\psi}$ be a Haar-random unit vector in $\mathbb R^{2^m}$, and for every $k\in \{0,1\}^n$, let $x_k=\|(\ketbra{k}{k}\otimes I)\ket{\psi}\|^2$. Then $$\sum_{k\in\{0,1\}^n}\left(\EE_{\ket{\psi}\sim\mu_{2^m}^\mathbb{R}}\left[x_k^2\right]\right)=\frac{2^{m-n}+2}{2^m+2}\,.$$
\end{lemma}
The proof of Lemma~\ref{lem:bound of Ex_i^2} can be found in Appendix~\ref{sec: proof for Ex_i^2}.

\paragraph{Bounding $\|\Pi(\rho-\sigma)\Pi\|_1$.} For every $k\in \{0,1\}^n$, let $\ket{\psi_k}=\Pi_k\ket{\psi}$,  $\ket{\hat{\psi}_k}=(I-\Pi_k)\ket{\psi}$ and $x_k=\|\ket{\psi_k}\|^2$. Our proof will rely on the fact that $x_k$ and the directions of $\ket{\psi_k}$ and $\ket{\hat\psi_k}$ are pairwise independent. So let $\ket{\hat{\psi}_k^-}=(I-\Pi_k^-)\ket{\psi^-}$, by the definition of $\Pi$ (Equation~\eqref{eq:def of Pi}), we have that
    $$\Pi\rho\Pi=\EE_{\ket{\psi}\sim \mu_{2^m}^{\mathbb R}}\left[\sum_{k\in\{0,1\}^n}\ketbra{\psi_k}{\psi_k}\otimes \ketbra{\hat{\psi}_k^-}{\hat{\psi}_k^-}^{\otimes t}\right]=\sum_{k\in\{0,1\}^n}\frac{\Pi_k}{2^{m-n}}\otimes \left(\EE_{\ket{\psi}\sim \mu_{2^m}^{\mathbb R}}\left[ x_k\ketbra{\hat{\psi}_k^-}{\hat{\psi}_k^-}^{\otimes t}\right]\right)\,,$$
    and we also have that
    $$\Pi\sigma\Pi=\EE_{\ket{\psi}\sim \mu_{2^m}^{\mathbb R}}\left[\sum_{k\in\{0,1\}^n}\frac{\Pi_k}{2^m}\otimes \ketbra{\hat{\psi}_k^-}{\hat{\psi}_k^-}^{\otimes t}\right]=\sum_{k\in\{0,1\}^n}\frac{\Pi_k}{2^m}\otimes \left(\EE_{\ket{\psi}\sim \mu_{2^m}^{\mathbb R}}\left[ \ketbra{\hat{\psi}_k^-}{\hat{\psi}_k^-}^{\otimes t}\right]\right)\,.$$
    Hence,
    \begin{align}
        \|\Pi(\rho-\sigma)\Pi\|_1&=
        \left\Vert\sum_{k\in\{0,1\}^n}\frac{\Pi_k}{2^{m-n}}\otimes \left(\EE_{\ket{\psi}\sim \mu_{2^m}^{\mathbb R}}\left[ \left(x_k-\frac{1}{2^n}\right)\ketbra{\hat{\psi}_k^-}{\hat{\psi}_k^-}^{\otimes t}\right]\right)\right\Vert_1\notag\\
        &\leq \sum_{k\in\{0,1\}^n}\left\Vert\EE_{\ket{\psi}\sim \mu_{2^m}^{\mathbb R}}\left[ \left(x_k-\frac{1}{2^n}\right)\ketbra{\hat{\psi}_k^-}{\hat{\psi}_k^-}^{\otimes t}\right]\right\Vert_1\label{eq:bound Pi(rho-sigma)Pi 1}\,,
    \end{align}
    where~\eqref{eq:bound Pi(rho-sigma)Pi 1} follows from a triangle inequality and the sub-multiplicativity of the trace norm. 
    
    Recall that $\ket{\hat \psi_k^-}=(\ket{0^{m+1}}-\ket{1}\ket{\hat \psi_k})/\sqrt{2}$. In order to remove the correlation between $x_k$ and $\ket{\hat \psi_k^-}$, let $\ket{\tilde \psi_k^-}=(\ket{0^{m+1}}-\ket{1}\ket{\hat \psi_k}/\|\ket{\hat \psi_k}\|)\sqrt{2}$. Then $\ket{\tilde \psi_k^-}$ only depends on the direction of $\ket{\hat \psi_k}$ (which is independent of $x_k$), so we can write that
    \begin{equation}\label{eq:E with tilde psi_k^- vanishes}
        \EE_{\ket{\psi}\sim \mu_{2^m}^{\mathbb R}}\left[ \left(x_k-\frac{1}{2^n}\right)\ketbra{\tilde{\psi}_k^-}{\tilde{\psi}_k^-}^{\otimes t}\right]=\EE_{\ket{\psi}\sim \mu_{2^m}^{\mathbb R}}\left[ \left(x_k-\frac{1}{2^n}\right)\right]\cdot \EE_{\ket{\psi}\sim \mu_{2^m}^{\mathbb R}}\left[ \ketbra{\tilde{\psi}_k^-}{\tilde{\psi}_k^-}^{\otimes t}\right]=0\,,
    \end{equation}
    since $\EE_{\ket{\psi}\sim \mu_{2^m}^{\mathbb R}} x_k=1/2^n$. In addition, notice that
    $$\left\Vert\ket{\hat \psi_k^-}-\ket{\tilde \psi_k^-}\right\Vert= \frac{1}{\sqrt{2}}\left\Vert\frac{\ket{\hat \psi_k}}{\|\ket{\hat \psi_k}\|}-\ket{\hat \psi_k}\right\Vert=\frac{1-\|\ket{\hat \psi_k}\|}{\sqrt{2}}=\frac{1-\sqrt{1-x_k}}{\sqrt{2}}\leq \frac{x_k}{\sqrt{2}}\,.$$
    Thus,
    \begin{align*}
        \left\Vert\ketbra{\hat \psi_k^-}{\hat \psi_k^-}-\ketbra{\tilde \psi_k^-}{\tilde \psi_k^-}\right\Vert_1 &= \left\Vert\left(\ket{\hat \psi_k^-}-\ket{\tilde \psi_k^-}\right)\bra{\hat \psi_k^-}+\ket{\tilde \psi_k^-}\left(\bra{\hat \psi_k^-}-\bra{\tilde \psi_k^-}\right)\right\Vert_1\notag \\
        &\leq \left\Vert\ket{\hat \psi_k^-}-\ket{\tilde \psi_k^-}\right\Vert\cdot \|\ket{\hat \psi_k^-}\|
        + \|\ket{\tilde \psi_k^-}\|\cdot \left\Vert\ket{\hat \psi_k^-}-\ket{\tilde \psi_k^-}\right\Vert\notag \\
        & \leq \frac{x_k}{\sqrt{2}}+\frac{x_k}{\sqrt{2}}=\sqrt{2}x_k\,,
    \end{align*}
    and using a hybrid argument we obtain that
    \begin{equation}\label{eq:tilde and hat are close}
        \left\Vert\ketbra{\hat \psi_k^-}{\hat \psi_k^-}^{\otimes t}-\ketbra{\tilde \psi_k^-}{\tilde \psi_k^-}^{\otimes t}\right\Vert_1\leq \sqrt{2}tx_k\,.
    \end{equation}
    Hence, Equation~\eqref{eq:bound Pi(rho-sigma)Pi 1} yields
    \begin{align}
        \|\Pi(\rho-\sigma)\Pi\|_1&\leq \sum_{k\in\{0,1\}^n}\left\Vert\EE_{\ket{\psi}\sim \mu_{2^m}^{\mathbb R}}\left[ \left(x_k-\frac{1}{2^n}\right)\left(\ketbra{\hat{\psi}_k^-}{\hat{\psi}_k^-}^{\otimes t}-\ketbra{\tilde{\psi}_k^-}{\tilde{\psi}_k^-}^{\otimes t}\right)\right]\right\Vert_1\label{eq:bound Pi(rho-sigma)Pi 1.5}\\
        &\leq \sum_{k\in\{0,1\}^n}\EE_{\ket{\psi}\sim \mu_{2^m}^{\mathbb R}}\left[\left(x_k+\frac{1}{2^n}\right)\left\Vert \ketbra{\hat{\psi}_k^-}{\hat{\psi}_k^-}^{\otimes t}-\ketbra{\tilde{\psi}_k^-}{\tilde{\psi}_k^-}^{\otimes t}\right\Vert_1\right]\label{eq:bound Pi(rho-sigma)Pi 2}\\
        &\leq \sum_{k\in\{0,1\}^n}\EE_{\ket{\psi}\sim \mu_{2^m}^{\mathbb R}}\left[\left(x_k+\frac{1}{2^n}\right)\cdot \sqrt{2} tx_k\right]\label{eq:bound Pi(rho-sigma)Pi 3}\\
        &=\sqrt{2}t\cdot \sum_{k\in\{0,1\}^n}\left(\EE_{\ket{\psi}\sim \mu_{2^m}^{\mathbb R}}\left[x_k^2\right]\right)+\frac{\sqrt{2}t}{2^n}\cdot \EE_{\ket{\psi}\sim \mu_{2^m}^{\mathbb R}}\left[\sum_{k\in\{0,1\}^n}x_k\right]\notag\\
        &=\sqrt{2}t\cdot \frac{2^{m-n}+2}{2^m+2}+\frac{\sqrt{2}t}{2^n}=O\left(\frac{t}{2^n}\right)\,,\label{eq:bound Pi(rho-sigma)Pi 4}
    \end{align}
    where~\eqref{eq:bound Pi(rho-sigma)Pi 1.5} follows from Equation~\eqref{eq:E with tilde psi_k^- vanishes},~\eqref{eq:bound Pi(rho-sigma)Pi 2} follows from the convexity of the trace norm,~\eqref{eq:bound Pi(rho-sigma)Pi 3} is due to Equation~\eqref{eq:tilde and hat are close}, and~\eqref{eq:bound Pi(rho-sigma)Pi 4} follows from Lemma~\ref{lem:bound of Ex_i^2} and the fact that $\sum_{k\in\{0,1\}^n}x_k=1$.

    \paragraph{Bounding $\|\rho-\Pi\rho\Pi\|_1$ and $\|\sigma-\Pi\sigma\Pi\|_1$.} As we discussed in Section~\ref{sec: tech overview}, a direct application of the gentle measurement lemma (Lemma~\ref{lem:gentle-measurement}) would allow us to bound both terms by $O(\sqrt{t/2^n})$, which would suffice for the proof of Theorem~\ref{thm:short-1PRS}. We defer the proof of this weaker bound in Appendix~\ref{sec:weak bounds}. 
    
    In fact, through a more careful treatment of the symmetries of the Haar measure, we will prove that both $\|\rho-\Pi\rho\Pi\|_1$ and $\|\sigma-\Pi\sigma\Pi\|_1$ are bounded by $O(t/2^n)$. 

    By the definition of $\Pi$ (Equation~\eqref{eq:def of Pi}), we have that
     \begin{align}
         \|\rho-\Pi\rho\Pi\|_1&=\left\Vert\EE_{\ket{\psi}\sim \mu_{2^m}^{\mathbb R}}\left[\sum_{k\in\{0,1\}^n}\ketbra{\psi_k}{\psi_k}\otimes \left(\ketbra{\psi^-}{\psi^-}^{\otimes t}-\ketbra{\hat{\psi}_k^-}{\hat{\psi}_k^-}^{\otimes t}\right)\right]\right\Vert_1\label{eq:bound rho-PirhoPi 1}\\
         &\leq \sum_{k\in\{0,1\}^n}\left\Vert\EE_{\ket{\psi}\sim \mu_{2^m}^{\mathbb R}}\left[\ketbra{\psi_k}{\psi_k}\otimes \left(\ketbra{\psi^-}{\psi^-}^{\otimes t}-\ketbra{\hat{\psi}_k^-}{\hat{\psi}_k^-}^{\otimes t}\right)\right]\right\Vert_1\,,\label{eq:bound rho-PirhoPi 2}
     \end{align}
     where~\eqref{eq:bound rho-PirhoPi 1} holds because the projectors $\{\Pi_k\otimes (I-\Pi_k^-)^{\otimes t}\}_{k\in\{0,1\}^n}$ are orthogonal, and~\eqref{eq:bound rho-PirhoPi 2} follows from a triangle inequality. Let $\ket{\zeta_k}=\ket{\psi^-}^{\otimes t}-\ket{\hat\psi_k^-}^{\otimes t}$, then we have that
     \begin{align}
          \ketbra{\psi^-}{\psi^-}^{\otimes t}-\ketbra{\hat{\psi}_k^-}{\hat{\psi}_k^-}^{\otimes t}&=\left(\ket{\zeta_k}+\ket{\hat{\psi}_k^-}^{\otimes t}\right)\left(\bra{\zeta_k}+\bra{\hat{\psi}_k^-}^{\otimes t}\right)-\ketbra{\hat{\psi}_k^-}{\hat{\psi}_k^-}^{\otimes t}\notag\\
          &=\ketbra{\zeta_k}{\zeta_k}+\ketbra{\zeta_k}{\hat\psi_k^-}^{\otimes t}+\ket{\hat\psi_k^-}^{\otimes t}\!\bra{\zeta_k}\,.\label{eq:introduce zeta_k}
     \end{align}
     Thus, Equation~\eqref{eq:bound rho-PirhoPi 2} yields
     \begin{equation}
         \|\rho-\Pi\rho\Pi\|_1\leq 
         \sum_{k\in\{0,1\}^n}
         \left\Vert\EE_{\ket{\psi}\sim \mu_{2^m}^{\mathbb R}}\left[\ketbra{\psi_k}{\psi_k}\otimes \left(\ketbra{\zeta_k}{\zeta_k}+\ketbra{\zeta_k}{\hat\psi_k^-}^{\otimes t}+\ket{\hat\psi_k^-}^{\otimes t}\!\bra{\zeta_k}\right)\right]\right\Vert_1\,.\label{eq:bound rho-PirhoPi 3}
     \end{equation}
     
     Recall that $\ket{\hat\psi_k^- }=(I-\Pi_k^-)\ket{\psi^-}$, so let $\ket{\psi_k^- }=\Pi_k^-\ket{\psi^-}=-\ket{1}\ket{\psi_k}/\sqrt{2}$, then we can write that $\ket{\psi^-}=\ket{\hat\psi_k^-}+\ket{\psi_k^- }$. Notice that $\ket{\zeta_k}=(\ket{\hat\psi_k^-}+\ket{\psi_k^- })^{\otimes t}-\ket{\hat\psi_k^-}^{\otimes t}$ is supported on vectors that contain at least one $\ket{\psi_k^-}$ factor, and let
     $$\ket{\phi_k}=\sum_{r=1}^t\ket{\hat\psi_k^-}^{\otimes (r-1)}\otimes \ket{\psi_k^-}\otimes \ket{\hat\psi_k^-}^{\otimes (t-r)}\,,$$
     then $\ket{\chi_k}=\ket{\zeta_k}-\ket{\phi_k}$ is supported on vectors that contain at least two $\ket{\psi_k^-}$ factors. Since the direction of $\ket{\psi_k}$ is independent of $\ket{\hat\psi_k}$, by the symmetry of the Haar measure we can replace $\ket{\psi_k}$ with $-\ket{\psi_k}$ (resp.\ $\ket{\psi_k^-}$ with $-\ket{\psi_k^-}$) to obtain that 
    \begin{align*}
        &\EE_{\ket{\psi}\sim \mu_{2^m}^{\mathbb R}}\left[\ketbra{\psi_k}{\psi_k}\otimes \ket{\phi_k}\!\bra{\hat\psi_k^-}^{\otimes t}\right]\\
        &=
        \sum_{r=1}^t\left(\EE_{\ket{\psi}\sim \mu_{2^m}^{\mathbb R}}\left[\ketbra{\psi_k}{\psi_k}\otimes \left(\ket{\hat\psi_k^-}^{\otimes (r-1)}\otimes \ket{\psi_k^-}\otimes \ket{\hat\psi_k^-}^{\otimes (t-r)}\right)\bra{\hat\psi_k^-}^{\otimes (t-r)}\right]\right)\\
        &=-\sum_{r=1}^t\left(\EE_{\ket{\psi}\sim \mu_{2^m}^{\mathbb R}}\left[\ketbra{\psi_k}{\psi_k}\otimes \left(\ket{\hat\psi_k^-}^{\otimes (r-1)}\otimes \ket{\psi_k^-}\otimes \ket{\hat\psi_k^-}^{\otimes (t-r)}\right)\bra{\hat\psi_k^-}^{\otimes (t-r)}\right]\right)\\
        &=-\EE_{\ket{\psi}\sim \mu_{2^m}^{\mathbb R}}\left[\ketbra{\psi_k}{\psi_k}\otimes \ket{\phi_k}\!\bra{\hat\psi_k^-}^{\otimes t}\right]\,,
    \end{align*}
     which implies that
     \begin{equation}\notag
         \EE_{\ket{\psi}\sim \mu_{2^m}^{\mathbb R}}\left[\ketbra{\psi_k}{\psi_k}\otimes \ket{\phi_k}\!\bra{\hat\psi_k^-}^{\otimes t}\right]=0\,.
     \end{equation}
     Similarly, we have that
     \begin{equation}\notag
         \EE_{\ket{\psi}\sim \mu_{2^m}^{\mathbb R}}\left[\ketbra{\psi_k}{\psi_k}\otimes \ket{\hat\psi_k^-}^{\otimes t}\!\bra{\phi_k}\right]=0\,.
     \end{equation}
     Hence, Equation~\eqref{eq:bound rho-PirhoPi 3} becomes
     \begin{align}
         \|\rho-\Pi\rho\Pi\|_1&\leq 
         \sum_{k\in\{0,1\}^n}
         \left\Vert\EE_{\ket{\psi}\sim \mu_{2^m}^{\mathbb R}}\left[\ketbra{\psi_k}{\psi_k}\otimes \left(\ketbra{\zeta_k}{\zeta_k}+\ketbra{\chi_k}{\hat\psi_k^-}^{\otimes t}+\ket{\hat\psi_k^-}^{\otimes t}\!\bra{\chi_k}\right)\right]\right\Vert_1\notag\\
         &\leq
         \sum_{k\in\{0,1\}^n}
         \EE_{\ket{\psi}\sim \mu_{2^m}^{\mathbb R}}\left[\left\Vert\ketbra{\psi_k}{\psi_k}\otimes \left(\ketbra{\zeta_k}{\zeta_k}+\ketbra{\chi_k}{\hat\psi_k^-}^{\otimes t}+\ket{\hat\psi_k^-}^{\otimes t}\!\bra{\chi_k}\right)\right\Vert_1\right]\,,\label{eq:bound rho-PirhoPi 4}
     \end{align}
     where~\eqref{eq:bound rho-PirhoPi 4} follows from the convexity of the trace norm. Recall that $\ket{\psi_k^- }=-\ket{1}\ket{\psi_k}/\sqrt{2}$, and we have that $\ket{\psi^-}=\ket{\hat\psi_k^-}+\ket{\psi_k^- }$, where $\ket{\psi_k^-}$, $\ket{\hat\psi_k^-}$ are orthogonal with $\|\ket{\psi_k^-}\|^2=x_k/2$ and $\|\ket{\hat\psi_k^-}\|^2=1-x_k/2$. Thus,
     \begin{align}
        \|\ket{\zeta_k}\|^2&=\sum_{r=1}^t\binom{t}{r}\|\ket{\hat\psi_k^-}\|^{2(t-r)}\|\ket{\psi_k^- }\|^{2r}\notag\\
        &=\sum_{r=1}^t\binom{t}{r}\left(1-\frac{x_k}{2}\right)^{t-r}\left(\frac{x_k}{2}\right)^{r}\notag\\
        &=1-\left(1-\frac{x_k}{2}\right)^t\leq \frac{tx_k}{2}\,,\label{eq:norm of zeta_k}
    \end{align}
    and
    \begin{align}
        \|\ket{\chi_k}\|^2&=\sum_{r=2}^t\binom{t}{r}\left(1-\frac{x_k}{2}\right)^{t-r}\left(\frac{x_k}{2}\right)^{r}\notag\\
        &= 1-\left(1-\frac{x_k}{2}\right)^t-\frac{tx_k}{2}\left(1-\frac{x_k}{2}\right)^{t-1}\notag \\
        &\leq  \frac{tx_k}{2}-\frac{tx_k}{2}\left(1-\frac{x_k}{2}\right)^{t-1}\label{eq:norm of chi_k 1}\\
        &\leq 
        \frac{tx_k}{2}\cdot \frac{(t-1)x_k}{2}\leq \frac{t^2x_k^2}{4}\label{eq:norm of chi_k 2}\,,
    \end{align}
    where~\eqref{eq:norm of zeta_k},~\eqref{eq:norm of chi_k 1} and~\eqref{eq:norm of chi_k 2} hold due to Bernoulli's inequality.
    
    Therefore,~\eqref{eq:bound rho-PirhoPi 3} yields
    \begin{align}
         \|\rho-\Pi\rho\Pi\|_1&\leq
         \sum_{k\in\{0,1\}^n}
         \EE_{\ket{\psi}\sim \mu_{2^m}^{\mathbb R}}\left[\|\ket{\psi_k}\|^2 \left(\|\ket{\zeta_k}\|^2+\|\ket{\chi_k}\|\cdot\|\ket{\hat\psi_k^-}\|^t+\|\ket{\hat\psi_k^-}\|^t\cdot\|\ket{\chi_k}\|\right)\right]\label{eq:bound rho-PirhoPi 5}\\
         &\leq \sum_{k\in\{0,1\}^n}
         \EE_{\ket{\psi}\sim \mu_{2^m}^{\mathbb R}}\left[x_k \left(\frac{tx_k}{2}+\frac{tx_k}{2}+\frac{tx_k}{2}\right)\right]\label{eq:bound rho-PirhoPi 6}\\
         &=\frac{3t}{2}\cdot \sum_{k\in\{0,1\}^n}\left(\EE_{\ket{\psi}\sim \mu_{2^m}^{\mathbb R}}\left[x_k^2\right]\right)\notag\\
        &=\frac{3t}{2}\cdot \frac{2^{m-n}+2}{2^m+2}=O\left(\frac{t}{2^n}\right)\label{eq:bound rho-PirhoPi 7}\,,
     \end{align}
     where~\eqref{eq:bound rho-PirhoPi 5} follows from the sub-multiplicativity of the trace norm,~\eqref{eq:bound rho-PirhoPi 6} follows from Equations~\eqref{eq:norm of zeta_k} and~\eqref{eq:norm of chi_k 2} and the fact that $\|\ket{\hat\psi_k^-}\|\leq 1$, and~\eqref{eq:bound rho-PirhoPi 7} is due to Lemma~\ref{lem:bound of Ex_i^2}.

     Similarly, for $\|\sigma-\Pi\sigma\Pi\|_1$ we have that
     \begin{align}
         \|\sigma-\Pi\sigma\Pi\|_1&=\left\Vert\EE_{\ket{\psi}\sim \mu_{2^m}^{\mathbb R}}\left[\sum_{k\in\{0,1\}^n}\frac{\Pi_k}{2^m}\otimes \left(\ketbra{\psi^-}{\psi^-}^{\otimes t}-\ketbra{\hat{\psi}_k^-}{\hat{\psi}_k^-}^{\otimes t}\right)\right]\right\Vert_1\label{eq:bound sigma-PisigmaPi 1}\\
         &\leq \sum_{k\in\{0,1\}^n}\left\Vert\EE_{\ket{\psi}\sim \mu_{2^m}^{\mathbb R}}\left[\frac{\Pi_k}{2^m}\otimes \left(\ketbra{\zeta_k}{\zeta_k}+\ketbra{\zeta_k}{\hat\psi_k^-}^{\otimes t}+\ket{\hat\psi_k^-}^{\otimes t}\!\bra{\zeta_k}\right)\right]\right\Vert_1\label{eq:bound sigma-PisigmaPi 2}\\
         &\leq \sum_{k\in\{0,1\}^n}\left\Vert\EE_{\ket{\psi}\sim \mu_{2^m}^{\mathbb R}}\left[\frac{\Pi_k}{2^m}\otimes \left(\ketbra{\zeta_k}{\zeta_k}+\ketbra{\chi_k}{\hat\psi_k^-}^{\otimes t}+\ket{\hat\psi_k^-}^{\otimes t}\!\bra{\chi_k}\right)\right]\right\Vert_1\label{eq:bound sigma-PisigmaPi 3}\\
         &\leq \sum_{k\in\{0,1\}^n}\EE_{\ket{\psi}\sim \mu_{2^m}^{\mathbb R}}\left[\frac{2^{m-n}}{2^m}\cdot \left(\frac{tx_k}{2}+\frac{tx_k}{2}+\frac{tx_k}{2}\right)\right]\label{eq:bound sigma-PisigmaPi 4}\\
         &= \frac{3t}{2^{n+1}}\cdot \EE_{\ket{\psi}\sim \mu_{2^m}^{\mathbb R}}\left[\sum_{k\in\{0,1\}^n}x_k\right]=O\left(\frac{t}{2^n}\right)\label{eq:bound sigma-PisigmaPi 5}\,,
     \end{align}
    where~\eqref{eq:bound sigma-PisigmaPi 1} holds because $\sum_{k\in\{0,1\}^n}\Pi_k=I$,~\eqref{eq:bound sigma-PisigmaPi 2} follows from Equation~\eqref{eq:introduce zeta_k},~\eqref{eq:bound sigma-PisigmaPi 3} follows from the fact that $\EE_{\ket{\psi}\sim \mu_{2^m}^{\mathbb R}}\left[\ket{\phi_k}\!\bra{\hat\psi_k^-}^{\otimes t}\right]=\EE_{\ket{\psi}\sim \mu_{2^m}^{\mathbb R}}\left[\ket{\hat\psi_k^-}^{\otimes t}\!\bra{\phi_k}\right]=0$,~\eqref{eq:bound sigma-PisigmaPi 4} follows from the convexity and the sub-multiplicativity of trace norm and using Equations~\eqref{eq:norm of zeta_k} and~\eqref{eq:norm of chi_k 2}, and ~\eqref{eq:bound sigma-PisigmaPi 5} is due to $\sum_{k\in\{0,1\}^n}x_k=1$.

    The proof of Theorem~\ref{thm:real-haar-stat-close} is completed by substituting the bounds of Equations~\eqref{eq:bound Pi(rho-sigma)Pi 4},~\eqref{eq:bound rho-PirhoPi 7}, and~\eqref{eq:bound sigma-PisigmaPi 5} to Equation~\eqref{eq:bound by 3 terms}. 

    \begin{remark}
        Our attack for the longer 1PRS in Section~\ref{sec: oracle separation} shows that the bound $O(t/2^n)$ is tight in the following sense: for every $\epsilon,c>0$, the~\cite{AGL24} construction remains statistically secure as long as the adversary has access to at most $2^{(1-\epsilon)n}$ copies of the state $\ket{\psi^-}$, but an unbounded adversary can break the scheme using $2^{(1+c)n}$ copies of the state $\ket{\psi^-}$.
    \end{remark}

\section{A General Lifting Theorem for an Asymmetric Model}\label{sec: lifting theorem}
In this section, we introduce a general framework for lifting an asymmetric oracle separation to a standard oracle separation; we will formally define the notion of an asymmetric separation later. Informally, our lifting framework shows that even when the adversary in the oracle separation is granted strictly stronger oracle access than the generator, the resulting asymmetric separation can still be converted into a standard oracle separation.

There is a long line of works on oracle separations and their implications in feasibility of black-box constructions between cryptographic primitives, most of which work is in an abstract model of primitives, dating back to~\cite{RTV04}. However, the general lifting technique in this section would require an explicit notion of the ``advantage" of an adversary. Following the idea of~\cite{BH13}, we propose the following definition of cryptographic primitives, and work under this model throughout this section.

\begin{definition}[Cryptographic primitive]
\label{def:primitive}
A \textit{cryptographic primitive} $P$ consists of:
\begin{enumerate}
    \item For each security parameter $n$, a class
    $\mathcal F_{P,n}$ of admissible honest implementations,
    specifying the functionality and correctness requirements.
    An implementation may be a tuple of classical or quantum
    channels, or an interactive strategy.

    \item A security specification assigning an advantage
    \[
        \operatorname{Adv}_{P}(n,F_n,D_n)\in[0,1]
    \]
    to an implementation $F_n$ and an adversary $D_n$,
    through the prescribed security experiment.
\end{enumerate}
\end{definition}

\paragraph{Our terminology.} An \textit{implementation} of $P$ is a family $F=\{F_n\}_{n\in \mathbb N_+}$ with $F_n\in\mathcal F_{P,n}$.
We will call $F$ \textit{efficient} if its honest operations admit circuits of total size polynomial in $n$. We will additionally call it \textit{uniform} if these circuits can be generated by some polynomial time Turing machine on input $1^n$.

We call an implementation $F=\{F_n\}_{n\in \mathbb N_+}$ \textit{secure} against an adversary family $D = \{D_n\}_{n \in \mathbb N_+}$, if the function $\operatorname{Adv}_{P}
       (n,F_n,D_n)$ is negligible in $n$.

We say the primitive $P$ \textit{exists}, if there is an efficient uniform implementation of $P$ that is secure against all efficient nonuniform adversaries.

\paragraph{Implications of an asymmetric model.} The most natural oracle separation results follow the following approach: to separate cryptographic primitives $P$ and $Q$, we construct an oracle $A$, relative to which there exists an efficient construction of $P$ that is information-theoretically secure relative to $A$, while no information-theoretically secure $Q$ exists relative to $A$ and can be broken by a side oracle $B$. But to instantiate the oracle separation, we will incur a self-reference issue: the generator of $Q$ might also have access to the side oracle $B$, which might be not clear how to break with $B$ itself. To deal with the self-reference issue, \cite{HR04} introduces the following theorem.

\begin{theorem}[\cite{HR04}]
    Suppose there exist oracles $A,B$, such that
    \begin{itemize}
        \item There is an efficient $A$-oracle implementation of $P$ that is secure against all efficient $(A,B)$-oracle adversaries; while
        \item No efficient $A$-oracle implementation of $Q$ is secure against all efficient $(A,B)$-oracle adversaries.
    \end{itemize}
    Then there does not exist a fully black-box construction of $Q$ from $P$.
\end{theorem}

We prove the following theorem, upgrading the result in~\cite{HR04} to a standard oracle separation.

\begin{theorem}
    Fix oracles $A=\{A_n\}_{n\in \mathbb N_+},B=\{B_n\}_{n\in \mathbb N_+}$ (potentially classical or quantum unitary), and let $P,Q$ be cryptographic primitives following the syntax in Definition~\ref{def:primitive}. Assume:
    \begin{itemize}
        \item There is an efficient (and uniform) $A$-oracle implementation of $P$ that is secure against all efficient nonuniform $(A,B)$-oracle adversaries; while
        \item No efficient \emph{nonuniform} $A$-oracle implementation of $Q$ is secure against all efficient nonuniform $(A,B)$-oracle adversaries.\footnote{Furthermore, the adversaries should tolerate negligible error in the implementation of $Q$, if $B$ is quantum unitary.}
    \end{itemize}
    Then, there exists an oracle $\widehat B$, such that under $\mathcal O=(A,\widehat B)$, $P^{\mathcal O}$ exists, while $Q^{\mathcal O}$, and even nonuniform $Q^{\mathcal O}$, does not.
\end{theorem}

\begin{proof}
    Without loss of generality, we assume $A,B$ are both quantum unitary oracles, and for every $n$, $A_n,B_n$ are unitaries on $n$ qubits.
    We first prove the following lemma:

    \begin{lemma}\label{lem:nonuniform-threshold}
        Suppose no efficient nonuniform $A$-oracle implementation of $Q$ is secure against all efficient nonuniform $(A,B)$-oracle adversaries. Then for every $c>0$, there exists $r=r_c>0$ and an infinite set of integers $W_c\in \mathbb N_+$, such that for every nonuniform $A$-oracle algorithm $C=\{C_n\}_{n\in \mathbb N_+}$ with $\operatorname{SIZE}(C_n)\leq n^c$, there exists an $(A,B)$-oracle adversary $D=\{D_n\}_{n\in \mathbb N_+}$, satisfying that for all $n\in W_c$:
        $$\operatorname{SIZE}(D_n)\leq n^{r}\quad\text{and}\quad\operatorname{Adv}_Q(n,C_n,D_n)\geq n^{-r}\,.$$
    \end{lemma}
    \begin{proof}
        Assume that there exists $c>0$ that does not satisfy the statement of the lemma. Then, for every $r\in \NN_+$, there exists some threshold $N_r\in \mathbb N_+$ and some nonuniform algorithm $C_r=\{C_{r,n}\}_{n\in \mathbb N}$ with $\operatorname{SIZE}(C_{r,n})\leq n^c$, such that for every $n\geq N_r$ and every $D_n$ with $\operatorname{SIZE}(D_n)\leq n^r$, we have that
        $$\operatorname{Adv}_Q(n,C_{r,n},D_n)\leq n^{-r}\,.$$
        Choose $N_r$ strictly increasing (with respect to $r$) and set
        $$r(n)=\max\{r:n\geq N_r\}\,.$$
        Let
        $$ C_n'=C_{r(n),n}\quad\text{and}\quad C'=\{C_{n}'\}_{n\in \mathbb N_+}\,.$$
        It is clear that the size of every $C_n'$ is bounded by $n^c$. We will show that $C'$ is secure against all nonuniform adversaries. Indeed, for every $r\in \NN_+$, every $n\geq N_r$, and every $D_n$ with $\operatorname{SIZE}(D_n)\leq n^r$, we have that
        $$\operatorname{Adv}_Q(n,C_n',D_n)\leq n^{-r}\,.$$
        Rearranging the quantifiers, we obtain that for every $r_0\in \NN_+$ and every $\{D_n\}_{n\in \NN_+}$ such that for every $n\geq n_0$, $\operatorname{SIZE}(D_n)\leq n^{r_0}$, we have that
        $$\forall r\geq r_0,\,\forall n\geq \max\{N_r,n_0\}:\, \operatorname{Adv}_Q(n,C_n',D_n)\leq n^{-r}\,.$$
        This shows $\operatorname{Adv}_Q(n,C_n',D_n)$ is negligible for any polynomial size adversary $D=\{D_n\}_{n\in \mathbb N_+}$, which contradicts the assumption that no efficient nonuniform implementation of $Q$ exists. This concludes the proof of Lemma~\ref{lem:nonuniform-threshold}.
    \end{proof}

    We start by the construction of the oracle $\widehat B$. Choose a function $\kappa:\mathbb N_+\to \mathbb N_+$ such that for every $i\in \mathbb N_+$ there are infinitely many $s\in \mathbb N_+$ such that $\kappa(s)=i$. Let $R_0=2$. We recursively define the sequences $\{R_s\}_{s\in \NN}$ and $\{\Lambda_s\}_{s\in \NN^+}$ such that:
    $$\Lambda_s=\min\{\lambda\geq 2^{R_{s-1}}:\lambda\in W_{\kappa(s)}\}\quad\text{and}\quad R_s=\Lambda_s^{2{\kappa(s)}}+\Lambda_s^{r_{\kappa(s)}}\,.$$
    For every $j,s\in \NN_+$ such that $j\leq \Lambda_s^{r_{\kappa(s)}}$, let
    $$\widehat B_{\Lambda_s^{2\kappa(s)}+j}=B_j\otimes I_{\Lambda_s^{2\kappa(s)}}\,,$$
    and for all other indices, $\widehat B$ acts trivially.
    
    Clearly, any $(A,\widehat B)$-oracle circuit can be compiled into an $(A,B)$-oracle circuit of the same size. Thus, the existence of $P$ is automatic.
    
    We now focus on the nonexistence of $Q$. For any nonuniform $(A,\widehat B)$-oracle algorithm $G=\{G_n\}_{n\in \mathbb N_+}$ such that $\operatorname{SIZE}(G_n)\leq n^i$, let $S=\kappa^{-1}(2i)$. Choose $i$ to be sufficiently large. We break the security of $G$ at the security parameter $\Lambda_s$, for any sufficiently large $s\in S$.

    First, it is not hard to notice the following: for all $R_{s-1}<j\leq \Lambda_s^{2\kappa(s)}$, $\widehat B_j$ acts trivially. Since the size of $G_{\Lambda_s}$ is at most $\Lambda_s^{i}\leq \Lambda_s^{\kappa(s)}$, we can simply replace the queries to $\widehat B_j$ trivially, as long as $R_{s-1}<j\leq \Lambda_s^{2\kappa(s)}$.

    Next, for all $j\leq R_{s-1}$, recall that $\Lambda_s\geq 2^{R_{s-1}}$. So,
    $$j\leq R_{s-1}\leq \log_2 \Lambda_s\,.$$
    Thus, each query to $\widehat B_j$ can be replaced by a circuit of size $2^{cj}\leq \Lambda_s^c$ for some constant $c$ (from the Solovay--Kitaev Theorem).

    Thus, the $(A,\widehat B)$-oracle circuit $G_{\Lambda_s}$ can be compiled into an $A$-oracle circuit $\widetilde G_{\Lambda_s}$ of size $\Lambda_s^{i+c}\leq \Lambda_s^{2i}$. Define $\widetilde G$ as the nonuniform algorithm that runs $\widetilde G_{\Lambda_s}$ for each security parameter $\{\Lambda_s:s\in \kappa^{-1}(2i)\}$, and acts arbitrarily otherwise\footnote{That is, $\widehat G_n$ can be any implementation in $\mathcal F_{Q,n}$, as long as $n\notin \{\Lambda_s:s\in \kappa^{-1}(2i)\}$.}.
    
    Apply Lemma~\ref{lem:nonuniform-threshold} on $\widetilde G$: there exists a nonuniform $(A,B)$-oracle algorithm $D$, such that for any $n\in W_{2i}$,
    $$\operatorname{SIZE}(D_n)\leq n^{r_{2i}},\,\operatorname{Adv}_Q(n,\widetilde G_n,D_n)\geq n^{-r_{2i}}\,.$$
    Crucially, $\Lambda_s\in W_{\kappa(s)}=W_{2i}$. Therefore,
    $$\operatorname{SIZE}(D_{\Lambda_s})\leq \Lambda_s^{r_{\kappa(s)}},\,\operatorname{Adv}_Q(\Lambda_s,\widetilde G_{\Lambda_s},D_{\Lambda_s})\geq n^{-r_{2i}}\,.$$
    And finally, we convert the $(A,B)$-oracle circuit $D_{\Lambda_s}$ into an $(A,\widehat B)$-oracle circuit $\widehat D_{\Lambda_s}$: since $\operatorname{SIZE}(D_{\Lambda_s})\leq \Lambda_s^{r_{\kappa(s)}}$, we can simply replace each query $B_j$ with querying $\widehat B_{j+\Lambda_s^{2\kappa(s)}}=\widehat B_{j+\Lambda_s^{4i}}$. The size of $\widehat D_{\Lambda_s}$ will be at most $\Lambda_s^{r_{2i}+4i}$, whose exponent is a constant depending only on $i$.

    Thus, $\widehat D=\{\widehat D_{\Lambda_s}\}$ achieves an advantage of $n^{-r_{2i}}$ at the security parameters $\Lambda_s$ for each $s\in \kappa^{-1}(2i)$. Since $\kappa^{-1}(2i)$ is an infinite set, $\widehat D$ achieves a non-negligible advantage towards $G$.
\end{proof}

\ifnum\llncs=1
\bibliographystyle{splncs04}
\else
\bibliographystyle{alpha}
\fi
\bibliography{main}

\appendix
\section{Proof of Lemmas in Section~\ref{sec:short-1PRS-security}}\label{sec:Appendix}
\begin{proof}[Proof of Lemma~\ref{Z^k is equiv to measurement}]
    Let
    $$\ket{\psi}=\sum_{i\in \{0,1\}^m}c_i\ket{i}\,.$$
    Then
    \begin{align}
        &\EE_{k\gets\{0,1\}^n}(Z^k\otimes I)\ketbra{\psi}{\psi}(Z^k\otimes I)\notag\\
        ={}& \EE_{k\gets\{0,1\}^n}\left[\sum_{\substack{a,a'\in \{0,1\}^n\\ b,b'\in \{0,1\}^{m-n}}}c_{a,b}c_{a',b'}(Z^k\otimes I)\ketbra{a,b}{a',b'}(Z^k\otimes I)\right]\notag\\
        ={} & \EE_{k\gets\{0,1\}^n}\left[\sum_{\substack{a,a'\in \{0,1\}^n\notag\\ b,b'\in \{0,1\}^{m-n}}} (-1)^{k\cdot (a\oplus a')}c_{a,b}c_{a',b'}\ketbra{a,b}{a',b'}\right]\\
        ={} & \sum_{{\substack{a\in \{0,1\}^n\\ b,b'\in \{0,1\}^{m-n}}}} c_{a,b}c_{a,b'}\ketbra{a,b}{a,b'}\notag \\
        ={} & \sum_{k\in\{0,1\}^n}\left[\sum_{\substack{a,a'\in \{0,1\}^n\\ b,b'\in \{0,1\}^{m-n}}}c_{a,b}c_{a',b'}(\ketbra{k}{k}\otimes I)\ketbra{a,b}{a',b'}(\ketbra{k}{k}\otimes I)\right]\notag\\
        ={}& \sum_{k\in\{0,1\}^n}(\ketbra{k}{k}\otimes I)\ketbra{\psi}{\psi}(\ketbra{k}{k}\otimes I)\notag\,.
    \end{align}
\end{proof}

\label{sec: proof for Ex_i^2}
\begin{proof}[Proof of Lemma~\ref{lem:bound of Ex_i^2}]
Let $\{X_i\}_{i=1}^{2^m}\sim \mathcal{N}(0,1)$, and $\ket{X}=\sum_{i=1}^{2^m}X_i\ket{i}$. Set $\ket{\psi}=\ket{X}/\|\ket{X}\|$ and let $k\in \{0,1\}^n$. Then, we have that
    $$x_k=\frac{\|(\ketbra{k}{k}\otimes I)\ket{X}\|^2}{\|\ket{X}\|^2}=\frac{Y_k}{R}\,,$$
    where $Y_k=\sum_{\ell\in \{0,1\}^{m-n}}X_{k,\ell}^2$ and $R=\|\ket{X}\|^2=\sum_{i\in\{0,1\}^m}X_{i}^2$. Notice that $Y_k/R$ only depends on the direction of $\ket{\psi}$ and is independent of $R$. Thus, we can write that
    $$\EE\left[Y_k^2\right]=\EE\left[\frac{Y_k^2}{R^2}\right]\cdot\EE\left[R^2\right]=\left(\EE_{\ket{\psi}\sim\mu_{2^m}^\mathbb{R}}\left[x_k^2\right]\right)\cdot\EE\left[R^2\right]\,.$$
    Since $Y_k$ and $R$ are sums of squares of independent standard Gaussians, using Lemma~\ref{lem:chi^2 moment} we have that
    $$\EE_{\ket{\psi}\sim\mu_{2^m}^\mathbb{R}}\left[x_k^2\right]=\frac{\EE\left[Y_k^2\right]}{\EE\left[R^2\right]}=\frac{2^{m-n}(2^{m-n}+2)}{2^{m}(2^{m}+2)}=\frac{2^{m-n}+2}{2^{n}(2^{m}+2)}\,.$$
    Taking the summation over all $k\in\{0,1\}^n$, we obtain that $$\sum_{k\in\{0,1\}^n}\left(\EE_{\ket{\psi}\sim\mu_{2^m}^\mathbb{R}}\left[x_k^2\right]\right)=2^n\cdot \frac{2^{m-n}+2}{2^{n}(2^{m}+2)}=\frac{2^{m-n}+2}{2^{m}+2}\,.$$
\end{proof}

\section{Weaker Bounds for $\|\rho-\Pi\rho\Pi\|_1$ and $\|\sigma-\Pi\sigma\Pi\|_1$ via the Gentle Measurement Lemma}\label{sec:weak bounds}
It is easy to see that $\|\ket{\hat\psi_k^-}\|^2=(1+\|\ket{\hat\psi_k}\|^2)/2=1-x_k/2$. Hence, we can directly calculate
\begin{align}
    \Tr(\Pi\rho)&=\EE_{\ket{\psi}\sim \mu_{2^m}^{\mathbb R}}\left[\sum_{k\in\{0,1\}^n}\Tr\left(\ketbra{\psi_k}{\psi_k}\otimes \ketbra{\hat{\psi}_k^-}{\hat{\psi}_k^-}^{\otimes t}\right)\right]\notag\\
    &=\EE_{\ket{\psi}\sim \mu_{2^m}^{\mathbb R}}\left[\sum_{k\in\{0,1\}^n}\left(\|\ket{\psi_k}\|^2\cdot  \|\ket{\hat\psi_k^-}\|^2\right)\right]\notag\\
    &=\EE_{\ket{\psi}\sim \mu_{2^m}^{\mathbb R}}\left[\sum_{k\in\{0,1\}^n}x_k\left(1-\frac{x_k}{2}\right)^t\right]\notag\\
    &\geq \EE_{\ket{\psi}\sim \mu_{2^m}^{\mathbb R}}\left[\sum_{k\in\{0,1\}^n}x_k\left(1-\frac{tx_k}{2}\right)\right]\label{eq:Tr(Pirho) 1}\\
    &= \EE_{\ket{\psi}\sim \mu_{2^m}^{\mathbb R}}\left[\sum_{k\in\{0,1\}^n}x_k\right]-\frac{t}{2} \sum_{k\in\{0,1\}^n}\left(\EE_{\ket{\psi}\sim \mu_{2^m}^{\mathbb R}}\left[x_k^2\right]\right)\notag\\
    &=1-\frac{t}{2}\cdot \frac{2^{m-n}+2}{2^m+2}=1-O\left(\frac{t}{2^n}\right)\,,\label{eq:Tr(Pirho) 2}
\end{align}
where~\eqref{eq:Tr(Pirho) 1} follows from Bernoulli's inequality, and~\eqref{eq:Tr(Pirho) 2} follows from Lemma~\ref{lem:bound of Ex_i^2} and the fact that $\sum_{k\in\{0,1\}^n}x_k=1$.

Hence, $\Tr((I-\Pi)\rho)=O(t/2^n)$, so using Lemma~\ref{lem:gentle-measurement}, we obtain that $\|\rho-\Pi\rho\Pi\|_1\leq O(\sqrt{t/2^n})$.

Similarly, 
    \begin{align}
    \Tr(\Pi\sigma)&=\EE_{\ket{\psi}\sim \mu_{2^m}^{\mathbb R}}\left[\sum_{k\in\{0,1\}^n}\Tr\left(\frac{\Pi_k}{2^{m}}\otimes \ketbra{\hat{\psi}_k^-}{\hat{\psi}_k^-}^{\otimes t}\right)\right]\notag\\
    &=\EE_{\ket{\psi}\sim \mu_{2^m}^{\mathbb R}}\left[\sum_{k\in\{0,1\}^n}\frac{1}{2^n}\left(1-\frac{x_k}{2}\right)^t\right]\notag\\
    &\geq \EE_{\ket{\psi}\sim \mu_{2^m}^{\mathbb R}}\left[\sum_{k\in\{0,1\}^n}\frac{1}{2^n}\left(1-\frac{tx_k}{2}\right)\right]\label{eq:Tr(Pisigma) 1}\\
    &= 1-\frac{t}{2^{n+1}} \sum_{k\in\{0,1\}^n}\left(\EE_{\ket{\psi}\sim \mu_{2^m}^{\mathbb R}}\left[x_k\right]\right)=1-O\left(\frac{t}{2^n}\right)\,,\label{eq:Tr(Pisigma) 2}
\end{align}
where~\eqref{eq:Tr(Pisigma) 1} follows from Bernoulli's inequality, and~\eqref{eq:Tr(Pisigma) 2} holds because $\sum_{k\in\{0,1\}^n}x_k=1$.

Thus, $\Tr((I-\Pi)\sigma)=O(t/2^n)$, so using Lemma~\ref{lem:gentle-measurement}, we obtain that $\|\sigma-\Pi\sigma\Pi\|_1\leq O(\sqrt{t/2^n})$.

\end{document}